\documentclass[english, 12pt]{article}
\usepackage{natbib, graphicx, amsmath, amssymb,mathrsfs,amsfonts,url}
\usepackage[margin=1in]{geometry}
\usepackage{amsthm} 
\usepackage{cleveref, autonum}
\usepackage{color, soul}
\usepackage{xcolor}
\usepackage{rotating} 
\usepackage{enumitem}
\usepackage{pdfpages}
\usepackage{multirow}
\usepackage[framemethod=tikz]{mdframed} 

\usepackage[T1]{fontenc}
\usepackage[utf8]{inputenc}
\usepackage{authblk}
\usepackage{babel}
\usepackage{caption}
\usepackage{tikz}
\usepackage{pgfplotstable}
\usepackage{pgfplots}
\usepackage[normalem]{ulem}

\newtheorem{theorem}{Theorem}[section]

\usepackage{amsmath}
\usepackage{times}
\usepackage{bm}
\usepackage[plain,noend]{algorithm2e}

\newcommand{\GP}{\mathrm{GP}}
\newcommand{\KL}{\mathrm{KL}}
\newcommand{\KLD}{Kullback-Leibler divergence}
\newcommand{\tr}{\mathrm{tr}}

\newcommand{\BF}{\mathrm{BF}}
\newcommand{\pr}{\mathrm{pr}}
\newcommand{\MM}{\mathcal{M}}
\newcommand{\MN}{\mathcal{N}}
\newcommand{\wKL}{\widetilde{\KL}}

\newcommand{\iid}{independent and identically distributed }
\newcommand{\bandwidth}{\lambda} %% or {\alpha}
\newcommand{\RKHS}{\mathbb{H}^{\bandwidth}}
\newcommand{\Holder}{\alpha}

\title{Comparing and weighting imperfect models using D-probabilities}

\author[1]{Meng Li}
\author[2]{David B. Dunson}
\affil[1]{Department of Statistics, Rice University}
\affil[2]{Department of Statistical Science, Duke University}
\begin{document}\sloppy 
\maketitle 
\begin{abstract}
We propose a new approach for assigning weights to models using a divergence-based method ({\em D-probabilities}), relying on evaluating parametric models relative to a nonparametric Bayesian reference using Kullback-Leibler divergence.  D-probabilities are useful in goodness-of-fit assessments, in comparing imperfect models, and in providing model weights to be used in model aggregation.  D-probabilities avoid some of the disadvantages of Bayesian model probabilities, such as large sensitivity to prior choice, and tend to place higher weight on a greater diversity of models.  In an application to linear model selection against a Gaussian process reference, we provide simple analytic forms for routine implementation and show that D-probabilities automatically penalize model complexity. Some asymptotic properties are described, and we provide interesting probabilistic interpretations of the proposed model weights. The framework is illustrated through simulation examples and an ozone data application.  
\end{abstract}

\noindent {\em Key words}: Gaussian process; Gibbs posterior; Kullback-Leibler divergence; Model aggregation; Model selection; M-open; Nonparametric Bayes; Posterior probabilities. 

\section{Introduction}

Dealing with uncertainty in model choice is one of the fundamental tasks in statistics~\citep{claeskens2008model}. Suppose we have a list of parametric models under consideration $\mathcal{M}  = \{\mathcal{M}_1,\ldots,\mathcal{M}_k\}$ for the observations $y^{(n)} = \{y_1, \ldots, y_n\} \in \mathcal{Y}^n$ with $\mathcal{Y}$ the sample space. Each model $\mathcal{M}_j$ has a corresponding likelihood $p(\cdot \mid \theta_j, \MM_j)$, with $\theta_j \in \Theta_j$ a finite dimensional parameter. 
Then it becomes of substantial interest to provide a {\em weight} on each model $\mathcal{M}_j$ to be used in goodness-of-fit assessments of model adequacy, for comparing model performance, and for aggregating different models targeted to prediction.  One of the most popular approaches is to use Bayesian model probabilities as weights, with these weights forming the basis of Bayesian Model Averaging (BMA). This article is motivated by an attempt to define weights that can improve upon %some of the deficiencies of 
Bayesian model probabilities.  

Assigning equal prior probabilities to each model for simplicity, and letting $\pi(\cdot \mid \mathcal{M}_j)$ denote the prior density for $\theta_j$, for $j=1,\ldots,k$, the posterior probability of model $\mathcal{M}_j$ is 
\begin{equation}
\label{eq:post1}
\pr( \mathcal{M}_j  \mid  y^{(n)} ) = \pi_j = \frac{ L_j( y^{(n)} ) }{ \sum_{l=1}^k L_l(y^{n}) },\quad \mbox{for $j=1,\ldots,k$}, 
\end{equation}
where $L_j( y^{(n)} ) = \int p(y^{(n)}  \mid  \theta_j, \MM_j) \pi( \theta_j  \mid  \mathcal{M}_j )d\theta_j$ is the marginal likelihood.  Philosophically, in order to interpret $\pr( \mathcal{M}_j  \mid  y^{(n)} )$ as a model {\em probability}, one must rely on the (arguably always flawed) assumption that one of the models in the list $\mathcal{M}$ is exactly true, known as the $\MM$-closed case.  
%In contrast, $\mathcal{M}$-open assume the true model is known and $\mathcal{M}$-complete assumes the true model is known but possibly to complex to use, so that it becomes of interest to select a model in $\MM$ due to simplicity, interpretability or computational tractability~\citep{Bernardo+Smith:94}.  
However, from a pragmatic perspective, one can use $\pr( \mathcal{M}_j  \mid  y^{(n)} )$ as a model weight, regardless of the question of interpretation. 
This pragmatic view is supported by the well known result that asymptotically for regular parametric models, the posterior probability on the model that is closest to the true data-generating model in Kullback-Leibler divergence converges to one. 

Unfortunately, as model weights, Bayesian model probabilities have some practical disadvantages.  They are not useful for assessing model adequacy in an {\em absolute} sense, and hence are not calibrated for goodness-of-fit assessments.  Instead, they provide a measure of model performance {\em relative} to the other models under comparison.  A poor model may be assigned a high probability when the competing models are very poor, while a good model may be assigned a low probability when there are many good and/or similar competing models.  In addition, Bayesian model probabilities suffer from large sensitivity to the choice of the prior $\pi( \theta_j  \mid  \mathcal{M}_j )$ without an agreed upon method of default prior specification~\citep{Liang+:08}. Usual non-informative priors used in parameter estimation under a given model are typically improper and cannot be used.  In practice we have observed a tendency of BMA to be {\em over confident} in weighting models - assigning weights that are too close to zero or one.  

One possibility is to consider Bayesian model selection from an $\mathcal{M}$-open or $\mathcal{M}$-complete perspective to allow the true model to fall outside of $\MM$; the $\mathcal{M}$-complete case assumes the true model is known but possibly too complex~\citep{Bernardo+Smith:94}. In these cases, one can formulate the model selection problem in a decision theoretic framework~\citep{Bernardo+Smith:94,GP+:09, Clyde+Iversen:13}, selecting the model in $\MM$ that maximizes expected utility. Expected utility can be approximated either via  cross-validation~\citep{Clyde+Iversen:13} or using a nonparametric prior~\citep{GP+Walker:05,GP+:09}. Cross validation is computationally intensive, and maximizing expected utility produces a single optimal model without uncertainty quantification or weights to be used in model aggregation.    

There is a rich literature on alternative methods for weighting models.  As an approximation to BMA weights, it is common to calculate the Bayesian Information Criterion ($\mbox{BIC}_j$) for each model $\mathcal{M}_j$, and then use weights proportional to $\exp(-\mbox{BIC}/2)$ \citep{kass1995reference, Hoeting1999}. Many authors have proposed to use cross validation to empirically estimate weights to be used in model aggregated predictions.  One example is the so-called {\em super learner} \citep{vanderLaan2007}.  If the focus is on aggregating models, or more broadly predictive algorithms, then it is possible to recast the problem as a two-stage linear regression - in the first stage one fits each of the predictive algorithms separately and obtains the corresponding estimated predictive values, while in the second stage these predictive values are used as predictors in a linear regression.  One can then exploit the rich toolbox of methods for fitting high-dimensional linear regression models to aggregate large numbers of predictive algorithms. \cite{Rigollet2012a} developed an exponential weighting method targeted to aggregation of sparse Gaussian regression models.  All of these methods are focused on providing weights for model aggregation, and are not useful for goodness-of-fit assessments of (absolute) model adequacy.  

We propose a simple definition of model weights that are calibrated in an {\it absolute} sense. To estimate these weights, we require knowledge of the oracle model that generated the data. Using a nonparametric Bayes surrogate for the oracle, we provide methods for estimation and inference. The proposed model weights provide assessment of model adequacy and goodness-of-fit, describe uncertainty in model selection, and are useful in model aggregration.  While the framework is broad, we focus primarily on comparing linear models using a Gaussian process surrogate. The framework reduces sensitivity to the price choice, and default prior specification including improper priors can be used as long as the posteriors under each model are proper. These advantages are verified by a comprehensive simulation study under univariate settings (Section~\ref{sec:simulation}) and an ozone data application involving multivariate predictors (Section~\ref{section:application}). 

Our notion of model weights has a connection to a range of concepts in the literature including~\cite{Boltzmann:1878}.  We establish various probabilistic interpretations using an explicit decision rule in the setting of hypothetical repeated experiments and p-values in Section~\ref{sec:interpretation}. The calibration and coherence of the new model weights may make the framework an appropriate foundation for a wide range of problems beyond linear model selection.

\section{Absolute and relative model weights} 
\label{sec:model.probabilities} 
 
\subsection{Definition of model weights} 
Let $\mathcal{N}^*$ be the oracle model which generated the data and $f^*$ be the corresponding density function.  Let $\KL(f, g) = \int f \log (f / g)$.  
For any model $\MM_j$ with density $f_j$, we define the following {\em absolute} model weights: 
\begin{equation}
	\label{eq:prob.true.marginal}
	\pi_j = \exp\{-n \KL(f^*, f_j)\}, \; j = 1, \ldots, k, 
\end{equation}
which equals the exponentiated negative \KLD\  between $\MM_j$ and the oracle model.  This definition is closely related to the notion of the {\em extent} of a distribution, which was introduced by \cite{Campbell:66} using the exponentiated entropy, with the relative entropy \KLD\ as a special case.  To our knowledge, this notion of extent has been overlooked outside of information theory. 

Under \eqref{eq:prob.true.marginal} $\pi_j \in (0,1)$ since $\KL(f^*, f_j)$ is always nonnegative. However, simply obeying this constraint does not make $\pi_j$ interpretable as a probability or useful as a basis of inference. One obtains a probabilistic interpretation of the model weights \eqref{eq:prob.true.marginal} in an absolute sense if $\pi_j$ corresponds to the probability of an appropriately chosen event that reflects the likelihood under $\MM_j$ relative to 
$\mathcal{N}^*$.  Indeed, in Section~\ref{sec:decision.rule} we show that $\pi_j$ is the probability of selecting $\MM_j$ based on a randomized decision rule that chooses $\MM_j$ in the absence of sufficient evidence in the data to distinguish $\MM_j$ from $\mathcal{N}^*$.

\subsection{Conditional model weights} 
The definition of $\pi_j$ in~\eqref{eq:prob.true.marginal} provides an absolute measure of adequacy of a specific model.  In quantifying the relative performance of different models in a pre-specified list $\MM$, and in aggregating these models to obtain an ensemble predictive algorithm, it is useful to define conditional model weights.  We define the conditional weight for model $\MM_j$ as 
\begin{equation} \label{eq:prob.conditional} 
\pi_{j \mid \MM} = \frac{ \pi_j }{ \sum_{l=1}^k \pi_l } = \frac{ \exp\big\{-n \KL(f^*, f_j) \big\} }
{ \sum_{l=1}^k \exp\big\{-n \KL(f^*, f_l)\}}, 
\end{equation}
which is simply the absolute weight for model $\MM_j$ divided by the sum of the corresponding weights for each of the models in $\MM$.

The weights in \eqref{eq:prob.conditional} can be used to compare alternative parametric models.  Equation \eqref{eq:prob.conditional} has the same form as the famous Boltzmann-Gibbs weights in statistical mechanics with unit inverse temperature, where $\KL(f^*, f_j)$ is the energy of model $j$.  By defining conditional model weights relative to other models in the list $\MM$, we obtain a direct alternative to posterior model probabilities used in Bayesian inferences.  We will later show that the weights in \eqref{eq:prob.conditional} are asymptotically equivalent to usual posterior model probabilities if $f^* = f_j$ for some $j \in \{1,\ldots,k\}$, so that the oracle model exactly corresponds to one of the candidate parametric models.  Although we view this assumption as unrealistic, this property is nonetheless reassuring.  

\subsection{Estimation of model weights: D-probabilities} 
\label{sec:D-Bayes}

The model weights $\pi_j$ and $\pi_{j \mid \MM}$ cannot be calculated directly, because the oracle model $f^*$ is unknown and models in $\MM$ typically contain unknown parameters.
To allow $f^*$ to be unknown, we introduce a nonparametric reference model $\mathcal{N}$, which can be considered to be sufficiently flexible to accurately approximate the oracle, with accuracy improving with sample size.  The nonparametric reference has density $f_0$ and parameter $\theta_0$.  
The absolute and conditional model weights given the model list $\MM$ become
\begin{equation} \label{eq:prob.reference} 
\pi_{j} = \exp\big\{-n \wKL(f_0, f_j) \big\} ,  \quad 
\pi_{j \mid \MM} = \frac{ \exp\big\{-n \wKL(f_0, f_j) \big\} }
{ \sum_{l=1}^k \exp\big\{-n \wKL (f_0, f_l)\}}, 
\end{equation}
where $\wKL(f_0, f_j)$ is an estimate of the Kullback-Leibler divergence between model $\MM_j$ and the reference model $\MN$.  

We propose the following two estimators: a posterior mean estimator
\begin{equation}
\label{eq:posterior.mean}
\wKL_1(f_0, f_j) = \int \int \KL\{f_0(\cdot \mid \theta_0), f_j(\cdot \mid \theta_j) \} \pi(\theta_j \mid y^{(n)}) \pi(\theta_0 \mid y^{(n)}) d \theta_j d\theta_0, 
\end{equation}
and an estimator based on posterior predictive densities 
\begin{equation}
\label{eq:posterior.predictive}
\wKL_2(f_0, f_j) = \KL(\widehat{f}_0, \widehat{f}_j ), \;\text{where}\; \widehat{f}_j(\cdot) = \int f_j(\cdot \mid \theta_j) \pi(\theta_j \mid y^{(n)}) d \theta_j,\; \text{and}\; j= 0, 1, \ldots, k.  
\end{equation}
These two estimators address the uncertainty of parameters $(\theta_j, \theta_0)$ differently: 
the posterior mean estimator uses the posterior mean of $\KL\{f_0(\cdot \mid \theta_0), f_j(\cdot \mid \theta_j) \}$, while the  posterior predictive estimator uses the \KLD\ between predictive densities of each model. As shown later in Section~\ref{sec:BF}, the two estimators have the same asymptotic behavior and converge to the minimum \KLD\ to the oracle model among $\theta_j \in \Theta_j$ under mild conditions.  In practice, one can use whichever approximation is most convenient, or even rely on a mixture of~\eqref{eq:posterior.mean} and~\eqref{eq:posterior.predictive}. 

%Therefore, the model probability $\pi_j$ based on either $\wKL_1(f_0, f_j) $ or $\wKL_2(f_0, f_j)$ have the same interpretation as discussed in Section~\ref{sec:model.probabilities}. 
%%However, the finite sample performances of the resulting model probabilities could be dramatically deviating from each other, especially because of the involvement of the factor $n$ in calculating $\pi_j$. 
%This validability of both $\wKL_1(f_0, f_j) $ or $\wKL_2(f_0, f_j)$ provides flexibility in real application. Depending on the context, one may use a mixture of~\eqref{eq:posterior.mean} and~\eqref{eq:posterior.predictive} to ease the computation. For example,  in the illustrative normal linear model selection in Section~\ref{sec:linear.model},  we calculate the posterior predictive densities conditional on the variance components and then obtain the \KLD by integrating out the variances according to their posterior distributions, leading to a mixture of~\eqref{eq:posterior.mean} and~\eqref{eq:posterior.predictive} estimator which can be also viewed as an approximation to~\eqref{eq:posterior.predictive}. 

We refer to the quantities in expression (\ref{eq:prob.reference}) as {\em D-probabilities}, as they provide a divergence-based alternative to Bayesian posterior model probabilities.  D-probabilities provide an absolute measure of model adequacy and goodness-of-fit and avoid large sensitivity to prior choice; 
both of these issues are notoriously poorly addressed in the Bayesian literature.  The main challenges in the use of D-probabilities include the need to choose a nonparametric reference model, and develop accurate approximation algorithms.  The nonparametric Bayes literature provides a rich menu of possibilities for $\mathcal{N}$, ranging from Dirichlet processes ~\citep{Ferguson:73} to Gaussian processes~\citep{Rasmussen+Williams:06}; for a review, refer to \cite{Hjort+:10}.  There is a rich literature showing that Bayesian nonparametric models often have appealing frequentist asymptotic properties, such as appropriate notions of consistency~\citep{schwartz1965bayes} and optimal rates of convergence \citep{van+van:09, Bhattacharya2014, Castillo2014, shen2015adaptive,Ghosal+van:17}. 

For simplicity in exposition and computational ease, we focus on normal linear models with a Gaussian process reference for the remainder of the article except for Section~\ref{sec:interpretation}.  In this case, conditional on covariance parameters, the \KLD\ between each parametric model and the nonparametric model can be calculated analytically, allowing us to rapidly conduct analyses and more easily study properties of the proposed model weights. There has been extensive study showing optimality properties 
of Gaussian process priors, such as rate adaptive behavior in nonparametric regression \citep{van+van:09}.

Although we focus on Bayesian machinery, one can estimate D-probabilities using any method that estimates $\KL(f_0, f_j)$.  Substantial work has focused on estimating the \KLD\ between two unknown densities based on samples from these densities \citep{Leonenko2008, Perez-Cruz2008, Bu2018}. Our setting is somewhat different, but the local likelihood methods of \cite{Lee2006} and the Bayesian approach of \cite{Viele2007a} can potentially be used, among others. On the other hand, our proposed estimator of $\KL(f_0, f_j)$ may be of independent interest and can be used in other contexts. 

\section{D-Bayes inference for linear models}
\label{sec:linear.model} 

\subsection{Analytical forms of D-probabilities} 
\label{sec:analytical.form} 
Let $\{(x_i, y_i): x_i = (x_{i1}, \ldots, x_{ip}) \in \mathbb{R}^p, y_i \in \mathbb{R} \}_{i = 1}^n$ be \iid observations following the model 
\begin{equation}
\label{eq:model}
y \mid x \sim N\{\mu(x),  \sigma^2\}, 
\end{equation}
where $x$ is a $p$-dimensional predictor and $y$ is a univariate response.  
Let $Y = (y_1, \ldots, y_n)$ and $X = (x_1^T, \cdots, x_n^T)^T$. Letting $j=0$ index the reference model $\mathcal{N}$ and $j=1,\ldots,k$ index the parametric models, we let $\mu_j(\cdot)$ and $\sigma_j^2$ denote the mean function and variance, respectively, for model $j$.  According to the chain rule of the \KLD, $\wKL_t(f_0, f_j)$ is equal to the \KLD\   between the conditional densities of $y$ given $x$ followed by an expectation with respect to the distribution of $x$.  We use the empirical distribution of $x$ for both $\wKL_1(f_0, f_j)$ and $\wKL_2(f_0, f_j)$, which for example means that the term in~\eqref{eq:posterior.mean} is calculated by $\KL\{f_0(\cdot \mid \theta_0), f_j(\cdot \mid \theta_j)\} = \sum_{i = 1}^n \KL\{f_0(\cdot \mid \theta_0, x_i), f_j(\cdot \mid \theta_j, x_i)\}/n.$

We use a Gaussian process prior for $\mu_0(\cdot)$, with 
$\mu_0(\cdot) \mid \sigma, \tau, \lambda \sim \GP\{ 0,  \sigma^2 k(\cdot, \cdot; \lambda, \tau)\}$ and covariance function  
\begin{equation}
\label{eq: se.gp} 
k(x_{i_1}, x_{i_2}; \lambda, \tau ) = \tau^2 \exp\left\{\sum_{j = 1}^{p} \frac{- (x_{i_1, j} - x_{i_2, j})^2}{2 \lambda_{j}^2}\right\},
\end{equation}
having predictor-specific bandwidth parameters $\lambda = (\lambda_1, \ldots, \lambda_p)^T$. For notational convenience, we suppress the dependence of the 
covariance function on $\lambda$ and $\tau$; estimation of these hyper-parameters is discussed in Section~\ref{section:hyper.parameter}. The prior distribution of $\sigma_0^2$ is specified as $p(\sigma_0^2) \propto 1/\sigma_0^2$. Let $K$ be the covariance matrix whose $(i,j)$th element is $k(x_i, x_j)$, and $\mu_0^{(n)} = \{\mu_0(x_1), \ldots, \mu_0(x_n)\}$ be the conditional mean vector.  Then the reference model $\mathcal{N}$ assumes $Y \mid  \mu_0^{(n)}, \sigma_0 \sim N\{\mu_0^{(n)}, \sigma_0^2 I_n\}$, with priors $\mu_0^{(n)} \mid \sigma_0 \sim N(0, \sigma_0^2 K)$ and $p(\sigma_0^2) \propto 1/\sigma_0^2$. Letting $H = (I +  K^{-1})^{-1} = K (K + I)^{-1}$, we have 
\begin{equation}
\mu_0^{(n)} \mid  X, Y, \sigma_0 \sim N( H Y,  \sigma_0^2 H), \quad 
\sigma_0^2 \mid X, Y\sim \text{IG}\left\{ \frac{n}{2}, \frac{1}{2} Y^T (I - H) Y\right\}. 
\end{equation}
	
For model $\mathcal{M}_j$, let $x_j$ be a $p_j$-dimensional sub-vector of $x$ and $\mu_j(x) = (1,x_j^T)\beta_j$, so that model $\mathcal{M}_j$ has parameters $\theta_j = (\beta_j,\sigma_j^2)$. Letting $X_j$ denote the corresponding design matrix including a column of ones, the mean vector is $\mu_j^{(n)} = \{\mu_j(x_1),\ldots,\mu_j(x_n)\}^T = X_j \beta_j$.  With the following prior distributions: 
\begin{equation}
\label{eq:prior.m.j}
\beta_j | \sigma_j^2 \sim N(0, \sigma_j^2 \Sigma_j), \quad p(\sigma^2_j) \propto 1/\sigma^2_j, 
\end{equation}
for some prior covariance matrix $\Sigma_j$, the posterior distributions are 
\begin{equation}
%  X_j \beta_j 
\mu_j^{(n)} \mid \sigma_j^2, X_j, Y \sim N(H_j Y, \sigma_j^2 H_j), \quad \sigma_j^2 \mid X_j, Y \sim 
{\text{IG}}\left\{ \frac{n}{2}, \frac{1}{2}Y^T (I - H_j) Y\right\}, 
\end{equation}
where $H_j = X_j ( X_j^T X_j + \Sigma_j^{-1})^{-1} X_j^T$. 

The posterior mean estimates of $\wKL_1(f_0, f_j)$ in~\eqref{eq:posterior.mean} are obtained as follows. Conditional on unknown parameters $(\beta_j, \sigma_j, \mu_0^{(n)}, \sigma_0)$, the \KLD\ between model $\MM_j$ and the reference model $\mathcal{N}$ is 
\begin{equation} 
\label{eq:KL.risk} 
\KL(f_0, f_j \mid  \beta_j, \sigma_j, \mu_0^{(n)}, \sigma_0) 
= \frac{1}{2} \left\{\frac{\sigma_0^2}{\sigma_j^2} + \frac{(X_j \beta_j - \mu_0^{(n)})^T  (X_j \beta_j - \mu_0^{(n)})}{n \sigma_j^2} - 1 +  \log \frac{\sigma_j^2}{\sigma_0^2} \right\}. 
\end{equation} 
Using the fact that $(X_j \beta_j - \mu_0^{(n)}) \mid \sigma_j, \sigma_0 \sim N\{ (H_j - H)Y, \sigma^2_j H_j + \sigma_0^2 H\}$, we can further marginalize out $\mu_0^{(n)}$ and $\beta_j$ to obtain 
\begin{align}
\KL(f_0, f_j \mid  \sigma_j, \sigma_0)   
& = E_{\beta_j, \mu_0^{(n)} \mid \sigma_j, \sigma_0} \{ \KL(f_0, f_j \mid  \beta_j, \sigma_j, \mu, \sigma, \lambda, \tau)  \}\\
= & \frac{1}{2} \left\{\frac{ \sigma_0^2}{\sigma_j^2} + \frac{Y^T (H_j - H)^2 Y + \tr(\sigma_j^2 H_j + \sigma_0^2 H)}{ n \sigma_j^2} - 1 +   \log \frac{\sigma_j^2}{\sigma_0^2} \right\}\\
\label{eq:lm.KL.posterior.mean.conditional} 
= & \frac{1}{2} \left\{\frac{Y^T (H_j - H)^2 Y }{ n \sigma_j^2} + \frac{ \{1 + \tr(H)/n\}\sigma_0^2}{\sigma_j^2} + \log \frac{\sigma_j^2}{\sigma_0^2} + \frac{\tr(H_j)}{n} - 1 \right\}. 
\end{align}
By further integrating out $\sigma_j$ and $\sigma_0$, we obtain that $\wKL_1(f_0, f_j) =  E_{\sigma_0, \sigma_j}\KL(f_0, f_j \mid  \sigma_j, \sigma_0) $ as
\begin{equation}
\label{eq:lm.KL.posterior.mean} 
\wKL_1(f_0, f_j)  = \frac{1}{n} (\mathcal{G}_{j, 1} + \mathcal{P}_{j, 1}), 
\end{equation}
where 
\begin{equation}
\mathcal{G}_{j, 1}  = \frac{n}{2} \left[ \frac{Y^T (H_j - H)^2 Y}{ Y^T(I - H_j) Y}   + \frac{\{\tr(H) + n\}Y^T(I - H)Y}{(n - 2) Y^T(I - H_j) Y } + \log \frac{Y^T(I - H_j) Y}{Y^T (I - H) Y}  - 1 \right], 
\end{equation}
and $\mathcal{P}_{j, 1} = \tr(H_j) / 2. $

We next obtain $\wKL_2(f_0, f_j)$. Conditional on the variance parameters, the posterior predictive densities evaluated at $\{x_1, \ldots, x_n\}$ under the reference model and model $\MM_j$ are $N\{HY, \sigma_0^2 (I + H)\}$ and $ N\{H_jY, \sigma^2_j (I + H_j)\}$, respectively. Therefore, the \KLD\ $\wKL(f_0, f_j)$ conditional on the variances $(\sigma_j^2, \sigma^2_0)$ is 
\begin{equation}
\label{eq:lm.KL.predictive.conditional} 
\begin{split} 
\wKL_2(f_0, f_j \mid \sigma_j, \sigma_0) = \frac{1}{2} \left[\frac{Y^T(H_j - H)^T (I + H_j)^{-1} (H_j- H)Y}{n \sigma_j^2} \right.\\ \left. + \frac{\sigma_0^2}{\sigma_j^2} \frac{\tr\{(I + H_j)^{-1} (I + H)\}}{n}   +   \log  \frac{\sigma_j^2}{\sigma_0^2} + \frac{1}{n} \log\frac{\det (I + H_j)}{\det ( I + H)} - 1 \right].
\end{split} 
\end{equation}
Integrating out the variance parameters $\sigma_0^2$ and $\sigma_j^2$ leads to 
\begin{equation}
\label{eq:lm.KL.predictive}
\wKL_2(f_0, f_j) = \frac{1}{n} (\mathcal{G}_{j, 2} + \mathcal{P}_{j, 2}), 
\end{equation}
where 
\begin{align}
\mathcal{G}_{j, 2} & =   
\frac{n}{2} \left[ \frac{Y^T(H_j - H)^T (I + H_j)^{-1} (H_j- H)Y}{Y^T(I - H_j)Y}  + \frac{Y^T(I - H)Y}{Y^T(I - H_j) Y} \frac{\tr\{(I + H_j)^{-1} (I + H)\}}{n - 2}  \right. \\ 
&  + \left.  \log  \frac{Y^T(I - H_j) Y}{Y^T(I - H)Y} - \log {\det ( I + H)} - 1\right], 
\end{align}
and $\mathcal{P}_{j, 2} = \log{\det (I + H_j)}  / 2.$ 

Hence, both the posterior mean estimator in~\eqref{eq:lm.KL.posterior.mean} and posterior predictive density estimator in~\eqref{eq:lm.KL.predictive} admit the decomposition of the form 
$(\mathcal{G}_{j, t} + \mathcal{P}_{j, t}) / n$ for $ t = 1, 2$.  Let the corresponding D-probabilities be $\pi_{j, t} = \exp(-\mathcal{G}_{j, t} - \mathcal{P}_{j, t})$.  The term $\mathcal{G}_{j, t}$ is the goodness-of-fit of model $\MM_j$ compared to the reference model and $\mathcal{P}_{j, t}$ is a penalty term on model complexity.  The trace of $H_j$ is commonly used as the degrees of freedom of model $\MM_j$, and the log determinant of the fitted covariance matrix $\log \det (I + H_j)$ introduces a penalty on the rank of the covariance matrix~\citep{Fazel+:03}. Unlike most model selection criteria in the literature, the D-probability $\pi_{j, t}$ is interpretable in an absolute sense for each candidate model, as discussed in Section~\ref{sec:interpretation}. Therefore, the expression $\mathcal{G}_{j, t}$ keeps any constant even when it is the same across all models. 

If we  use the flat prior where $\Sigma_j^{-1} = 0$, the matrix $H_j$ is idempotent and we thus have $\tr(H_j) = p_j + 1$ and $ \log \det (I + H_j) = (p_j + 1) \log2$. Consequently, the D-probabilities penalize model complexity by 
\begin{equation}
\label{eq:penalty.flat}
-\mathcal{P}_{j, 1} = \frac{1}{2}(p_j + 1),\quad -\mathcal{P}_{j, 2} = \frac{\log 2}{2} (p_j + 1). 
\end{equation}
When comparing two models $\MM_j$ and $\MM_{j'}$ where $j \neq j'$, the relative penalties on model complexity are the same as used in some existing criteria.  
Specifically,  the penalty term $\mathcal{P}_{j', 1} - \mathcal{P}_{j, 1} = (p_j - p_{j'})/2$ is used in the Akaike information criterion~\citep{Akaike1998, akaike1974new} and the pseudo-Bayes factor~\citep{Geisser+Eddy:79}, while  $\mathcal{P}_{j', 2} - \mathcal{P}_{j, 2} = \log 2 (p_j - p_{j'}) / 2$ is the penalty term in the posterior Bayes factor~\citep{Aitkin:91, Gelfand+Dey:94}. %The built-in complexity penalization of the proposed D-probabilities is further confirmed by simulations in Section~\ref{sec:simulation}. 

% sensitivity to priors 
%The proposed D-probabilities are not sensitive to the choices of prior distributions, due to the usage of posterior distributions in calculating $\wKL(f_0, f_j)$. 
%Unlike the Bayes factors, improper priors are applicable here for any parameter in the model, if the corresponding posterior distributions are proper. 

%For example, consider the popular Zellner's $g$ priors where $\Sigma_j = g(X_j^T X_j)^{-1}$ for some constant $g > 0$~\citep{Zellner:86}, then the matrix $H_j = g/(1 + g) \cdot X_j(X_j^T X_j)^{-1} X_j^T$. Either the unit information prior when $g = n$ (Kass and Wasserman, 95) or the benchmark prior when $g = \max(n, (p_j + 1)^2)$ (Fernández et al. (2001) ) produce similar $H_j$ matrix for moderate sample size $n$. 

\subsection{Selection of hyperparameters} 
\label{section:hyper.parameter} 
We estimate the parameters $(\lambda, \tau)$ by maximizing the log marginal likelihood $\log p(Y \mid  \lambda, \tau).$ Based on the log-likelihood of $Y$ conditional on $\{\mu_0^{(n)}, \sigma^2_0, \lambda, \tau\}$,  we first integrate out  $\mu_0^{(n)}$ to obtain 
\begin{align}
\label{eq:log.marginal}
\log p(Y \mid  \lambda, \tau, \sigma^2_0) & = -\frac{n}{2} \log(2\pi) - \frac{1}{2} \log |\sigma_0^2 K + \sigma_0^2 I| - \frac{1}{2} Y^T (\sigma_0^2 K + \sigma_0^2 I)^{-1} Y \\
& = -\frac{n}{2} \log(2\pi) - \frac{n}{2} \log \sigma_0^2  - \frac{1}{2} \log |K +  I| - \frac{1}{2 \sigma_0^2} Y^T ( K + I)^{-1} Y, 
\end{align}
and further integrate out $\sigma_0^2$, 
\begin{equation}
\label{eq:log.marginal.lambda.tau} 
\log p(Y \mid  \lambda, \tau )  = - \frac{1}{2} \log |K +  I|  -\frac{n}{2} \log \{Y^T ( K + I)^{-1} Y\} + \text{constant}. 
\end{equation}
Let $(\lambda_{\text{EB}}, \tau_{\text{EB}})$ be the empirical Bayes estimates maximizing equation~\eqref{eq:log.marginal.lambda.tau}. Then the D-probability of model $\MM_j$ is 
\begin{equation}
\pi_j^{\text{EB}} = \exp\{-n \wKL(f_0, f_j \mid \lambda_{\text{EB}}, \tau_{\text{EB}})\}. 
\end{equation}

To avoid conditioning on an empirical point estimate of $(\lambda, \tau)$, one may alternatively implement Markov chain Monte Carlo methods to draw posterior samples of the $(p + 1)$-dimensional parameter $(\lambda, \tau)$ based on the likelihood in~\eqref{eq:log.marginal.lambda.tau} and priors with positive supports such as gamma distributions. %as in (BDA, page 504). 
Let $(\lambda^{(1)}, \tau^{(1)}), \ldots, (\lambda^{(J)}, \tau^{(J)})$ be the posterior samples after burn-in, then the D-probability of model $\MM_j$ is 
\begin{equation}
\label{eq:mcmc} 
\pi_j^{\text{MCMC}} = \exp\left\{- \frac{n}{J} \sum_{i = 1}^J \wKL(f_0, f_j\mid \lambda^{(i)}, \tau^{(i)})\right\}. 
\end{equation}

\section{Asymptotic behavior} 
\label{sec:BF} 

In this section, we investigate the asymptotic behavior of the proposed $\wKL_t(f_0, f_j)$ in Sections~\ref{sec:D-Bayes} for linear models, and relate conditional D-probabilities with usual posterior model probabilities.  We consider a compact support for the covariates, which is taken as $[0, 1]^d$ without loss of generality.  Let $C[0, 1]^d$ be the space of continuous functions on $[0, 1]^d$. For a function $g: [0, 1]^d \rightarrow \mathbb{R}$, $x \in [0, 1]^d$ and $i = 1, \ldots, d$, let $g_i(\cdot|x)$ be a univariate function such that for any $s \in [0, 1]$, $g_i(s | x) = g(\ldots, x_{i - 1}, s, x_{i + 1}, \ldots)$, i.e., the $i$th element in the argument is replaced by $s$.  Let $\|g\|_{\infty} = \sup_{x \in [0, 1]^d} |g(x)|$ be the supremum norm of a function $g$, then the anisotropic H\"{o}lder space $C^{\Holder}[0, 1]^d$ indexed by a vector of positive numbers  $\Holder = (\Holder, \ldots, \Holder_d)$ contains all functions $g$ such that for some $L > 0$, 
\begin{equation}
\underset{x \in [0, 1]^d}{\sup} \left\{\sum_{j = 0}^{\lfloor \Holder_i \rfloor} \| D^j g_i(\cdot | x) \|_{\infty} +
\frac{\|D^{\lfloor \Holder_i \rfloor} g_i(y + h \mid x) - D^{\lfloor \Holder_i \rfloor} g_i(y \mid x) \|_{\infty}}{|h|^{\Holder_i - \lfloor \Holder_i \rfloor}} \right\}\leq L  
\end{equation}
for any $(y, h)$ such that $y \in [0, 1]$, $h > 0$, $y + h \in [0, 1]$ and $i = 1, \ldots, d$; here $\lfloor \cdot \rfloor$ is the floor function and $D^j$ is the $j$th derivative operator. 
Let $\Holder_0^{-1} = \sum_{i = 1}^d \Holder_i^{-1}$ be an exponent of global smoothness~\citep{Birge1986, Barron1999a, Hoffmann2002}. 

For each model $\MM_j$, we define  
\begin{equation}
\label{eq:delta.j} 
\theta_j^* = \underset{\theta_j \in \Theta_j}{\arg \min}{\;\KL\{f^*, f_j(\cdot \mid \theta_j)\}},\quad 
\delta_j = \KL\{f^*, f_j(\cdot \mid \theta_j^*)\}. 
\end{equation}
The parameter value $\theta_j^*$ is the so-called pseudotrue parameter~\citep{Bunke1998}. A usual condition of Bayesian nonparametric models is that $\delta_j=0$ for all $f^*$ in a large set of densities. Unless $f^*$ exactly follows the parametric model under consideration, we have $\delta_j>0$ in general for any parametric model.

%The posterior distribution of $\theta_j$ concentrate around the parameter value that minimizes $\KL\{f_j(\cdot \mid \theta_j), f_0\}$ under mild conditions, even when the model $\MM_j$ is misspecified~\citep{Berk:66, Kleijn+van:06}.  Let $\delta_j$ be the minimum \KLD between $f_j(\cdot \mid \theta_j)$ and $f_0$, i.e., 
As the sample size $n$ increases, the posterior measure for the density $f$ under the nonparametric model $\mathcal{N}$ will tend to concentrate in arbitrarily small Kullback-Leibler neighborhoods of the true data-generating model $f^*$.  In contrast, the posterior measure for $f$ under the parametric model $\mathcal{M}_j$ will tend to concentrate on the point in the parametric class having the minimal Kullback-Leibler divergence from $f^*$.  Heuristically, this type of behavior suggests that the proposed $\wKL_t(f_0, f_j)$ will tend to converge to the minimal \KLD\ within the support of model $\MM_j$ as $n$ increases. However, as an information criterion, the \KLD\ may behave erratically~\citep{Barron1998a}, and the individual convergence of $f_0$ and $f_j$ does not directly imply the convergence of $\wKL_t(f_0, f_j)$~\citep{lkeda1960}. We overcome these difficulties by taking advantage of the Gaussianity assumption on the errors, which allows us to relate the \KLD\ to well studied distances on model parameters.  This is formalized in Theorem \ref{thm:BNP} based on the following assumptions. 
\begin{enumerate}[label=(\alph*)]
	\item Let $\theta_0^* = \{\mu_0^*, \sigma_0^*\}$ be the true parameter values under model~\eqref{eq:model}, which satisfy $\mu_0^*(\cdot) \in C^{\Holder}[0, 1]^d$ and $\sigma_0^* \in [a, b] \subset [0, \infty)$.
	The covariate $x$ is either fixed or randomly drawn from a density on $[0, 1]^d$ that is bounded away from zero and infinity. 
	\item For the reference model, the regression function $\mu_0$ has a Gaussian process prior $\Pi_{\bandwidth}$ with the squared exponential kernel function as in~\eqref{eq: se.gp}; the prior distribution of $\sigma_0$ has continuous density and is supported on $[a, b]$. 
	\item For each candidate model $\MM_j$, the prior distribution of $\theta_j = (\beta_j, \sigma_j)$ is supported on $\Theta_j  = \mathbb{R}^{p_j + 1} 
	\times [a, b]$,  which has a continuous density that is bounded away from zero and infinity. 
	The pseudotrue parameter $\theta_j^*$ is unique and interior to $\Theta_j$. 
\end{enumerate}
 
\begin{theorem}
	\label{thm:BNP}
	Under model~\eqref{eq:model} and Assumptions (a), (b) and (c), for $\wKL_1(f_0, f_j)$ in~\eqref{eq:posterior.mean}, if $\bandwidth_i = n^{\Holder_0/(2 \Holder_0 \Holder_i + \Holder_i)} (\log n)^{(d + 1)\Holder_0/(2 \Holder_0 \Holder_i+ \Holder_i)}$,  there exits a universal constant $c > 0$ such that 
	\begin{equation}
	E_{\theta^*} | \wKL_1(f_0, f_j)  - \delta_j | \leq c n^{-\Holder_0/(2 \Holder_0 + 1)} (\log n)^{(d + 1)\Holder_0/(2 \Holder_0 + 1)}. 
	\end{equation}
	For $\wKL_2(f_0, f_j)$ in~\eqref{eq:posterior.predictive}, we have 
	\begin{equation}
	0 \leq E_{\theta_0^*}\{\wKL_2(f_0, f_j) - \delta_j \}\leq c n^{-\Holder_0/(2 \Holder_0 + 1)} (\log n)^{(d + 1)\Holder_0/(2 \Holder_0 + 1)}
	\end{equation}
	for sufficiently large $n$.  
\end{theorem}
\begin{proof}
	See the Appendix.  
\end{proof}
%\begin{remark}
The anisotropic H\"{o}lder space has been used to consider dimension-specific smoothness; for example, see~\cite{Barron1999a}. The rate $n^{-\Holder_0/(2 \Holder_0 + 1)}$ is the minimax rate of convergence for a function in $C^{\Holder}[0, 1]^d$ according to~\cite{Hoffmann2002}. We can obtain a rate-adaptive version of Theorem~\ref{thm:BNP} without requiring the knowledge of $\Holder$ to select $\bandwidth$ by introducing an appropriate hyper-prior on $\lambda$ following the random rescaling scheme in~\cite{van+van:09} and~\cite{Bhattacharya2014}. 

%\end{remark} 
% \begin{remark}
	 Furthermore, the Gaussian process prior in Assumption (b) can be replaced by any nonparametric priors that lead to nearly optimal contraction rate of the mean function under $\|\cdot\|_n$ or a stronger metric, such as random series priors using a wavelet basis~\citep{Castillo2014} or B-splines~\citep{Yoo2016}. 
%\end{remark}

Theorem~\ref{thm:BNP} suggests that the numerator in $\pi_j$ given by~\eqref{eq:prob.reference} is approximately  $\exp(-n \delta_j)$ for large $n$. Consequently,  for $j_1 \neq j_2$, the ratio $\pi_{j_1} / \pi_{j_2}$ approximates $\exp\{-n(\delta_{j_1} - \delta_{j_2})\}$. The commonly used Bayes factor between two candidate models has been proven to have the same asymptotic behavior. For any two different models $\mathcal{M}_{j_1}$ and $\mathcal{M}_{j_2}$, the Bayes factor $\BF(\mathcal{M}_{j_1}, \mathcal{M}_{j_2}) = I^{(j_1)}_{n} / I^{(j_2)}_{n}$ is approximately equal to $\exp\{-n(\delta_{j_1} - \delta_{j_2})\}$ under mild conditions~\cite[Theorem 1]{Walker+:04}, suggesting its asymptotic equivalence to the proposed D-probabilities. 

D-probabilities can provide evidence of lack of fit of a parametric model under consideration, and the connection with Bayes factors further endows D-probabilities with a scale for strength of evidence (of a lack of fit) by adopting the convention in Bayes factors~\citep{kass1995reference}: 
	\begin{center}
	\begin{tabular}{cc}
		D-probability & Strength of evidence (of a lack of fit) \\
		 $ < 1/150$ & very strong \\
		 $[1/150, 1/20)$ & strong \\
		 $[1/20, 1/3)$ & positive \\
		 $[1/3, 1]$ &  not worth more than a bare mention
	\end{tabular} 
\end{center}
However, unlike Bayes factors which do not allow improper priors on model-specific parameters, D-probabilities are well defined under improper priors as long as the posteriors under each model are proper.

In addition to the pragmatic advantages of the proposed D-Bayes inference framework, our notion of model weights has a connection to a range of concepts. In Section~\ref{sec:interpretation}, we establish various probabilistic interpretations to justify calibration and coherence.  

\section{Probabilistic interpretations of model weights} 
\label{sec:interpretation} 
\subsection{Relationship to Boltzmann's formulation} 
\label{sec:Boltzmann} 
We are interested in exploring probabilistic interpretations of  \eqref{eq:prob.true.marginal}.  We start by considering discrete sample spaces with $\mathcal{Y} = \{\mathcal{y}_1,\ldots,\mathcal{y}_m \}$.  The probability mass function $f^*$ under the oracle model places probability $a_l$ on element $\mathcal{y}_l$, for $l=1,\ldots,m$, while the probability mass function $f_j$ under model $\MM_j$ places probabilities $b_1,\ldots,b_m$ on these elements.  Under the oracle model, the expected number of occurrences of $\mathcal{y}_l$ in $n$ trials is $n a_l$, for $l=1,\ldots,m$; we refer to these values as the oracle frequencies.  The probability of obtaining these frequencies from $n$ independent observations from $f_j$ is 
\begin{equation}
\label{eq:multinomial}
\text{Multinomial}(n, na_1, \ldots, na_m; f_j) = {{n}\choose{n a_1, \ldots, na_m}} b_1^{n a_1} \cdot \cdots \cdot b_m^{n a_m}. 
\end{equation}
As commented by~\cite{Akaike:85}, 
% ~citep(A Celebration of statistics, THe ISIS centenary Volume, 1985, 1-24), 
~\cite{Boltzmann:1878} derived that the probability in~\eqref{eq:multinomial} is asymptotically equal to 
$\exp\{-n \KL(f^*, f_j)\}$ 
up to a multiplicative constant. Since $\KL(f^*, f^*) = 0$, %we obtain that
\begin{equation}
\label{eq:bolzmann}
\exp\{-n \KL(f^*, f_j)\} = \frac{\exp\{-n \KL(f^*, f_j)\} }{\exp\{-n \KL(f^*, f^*)\} } \approx \frac{\text{Multinomial}(n, na_1, \ldots, na_m; f_j)}{\text{Multinomial}(n, na_1, \ldots, na_m; f^*)}, 
\end{equation}
as $n \rightarrow \infty$.  The right hand side of (\ref{eq:bolzmann}) is interpretable as the likelihood of obtaining the oracle frequencies under model $\MM_j$ relative to the likelihood under the oracle model.  As the multinomial likelihood of the oracle frequencies is maximized under the oracle model $f^*$, the right hand side of (\ref{eq:bolzmann}) is between zero and one, with the value moving closer to one as model $\MM_j$ improves relative to the oracle.  

Although the Boltzmann (1987) probabilistic interpretation of \eqref{eq:prob.true.marginal} is specific to discrete distributions, the justification can be extended to continuous sample spaces $\mathcal{Y}$ by first partitioning $\mathcal{Y}$ into bins $\mathcal{Y}_1,\ldots,\mathcal{Y}_m$ having $\mathcal{Y} = \bigcup_{l=1}^m \mathcal{Y}_l$ and $\mathcal{Y}_l \bigcap \mathcal{Y}_{l'} = \emptyset$ for all $l \neq l'$.  Then, letting $m \to \infty$ with bin size $| \mathcal{Y}_l | \to 0$, one obtains a limiting form of (\ref{eq:bolzmann}).  Therefore, we can generally use the exponentiated entropy between a candidate model and true model as a type of absolute probability weight on the candidate model for both discrete and continuous distributions.

\subsection{Decision rules, model probabilities and p-values} 
\label{sec:decision.rule} 
% \subsubsection{Repeated experiments and the {\it Extent} of a distribution}
The proposed model weights can also be obtained by an explicit decision rule in the setting of hypothetical repeated experiments. Suppose we have $m$ repeated experiments ($t = 1, 2, \ldots, m$), where the observations $y_t^{(n)} = \{ y_{t1},\ldots, y_{tn} \}$ are drawn independently from $f^*$ and the different experiments are independent.  Define the likelihood ratio statistic for testing model $\MM_j$ against the oracle model using data from experiment $t$ as: 
\begin{equation}
T_{jt}^{(n)} = \frac{\prod_{i = 1}^n f_j(y_{ti})}{\prod_{i = 1}^n f^*(y_{ti})}.
\end{equation}
The geometric mean of these likelihood ratio statistics across repeated experiments is $R_{jm} = (\prod_{t = 1}^m T_{jt}^{(n)})^{1/m}.$  As $m$ increases, 
$R_{jm} \rightarrow \exp\{-n \KL(f^*, f_j)\}$ almost surely according to the strong law of large numbers, so for sufficiently large $m$, $R_{jm} < 1$ almost surely.  
%To avoid the possibility of $R_m>1$ for finite $m$, we let $R_m = \min(R_m,1)$.

We define a random decision rule in which $Z_{jm}=1$ corresponds to choosing model $\MM_j$ based on the data from $m$ replicated experiments, with $Z_{jm} = 0$ otherwise.  Choosing model $\MM_j$ is an absolute model selection decision about the merits of model $\MM_j$. 
%that does not depend on the quality of other models under consideration.  
Based on data from $m$ repeated experiments, as our decision rule we let 
\begin{equation}
Z_{jm} \sim \mbox{Bernoulli}( R_{jm} \wedge 1 ),
\end{equation}
where we take the minimum of $R_{jm}$ and one to remove the possibility of $R_{jm}>1$ for finite $m$.   This decision rule will tend to set $Z_{jm}=1$ with high probability if $f_j$ provides an accurate approximation to $f^*$, with accuracy judged relative to the sample size $n$; as sample size becomes larger it is appropriate to ask more of a parametric model.  If on average across experiments the information in data having a sample size of $n$ is sufficient to clearly distinguish the parametric and oracle model, then the decision rule will tend to set $Z_{jm}=0$ with high probability.  In such a case, model $\MM_j$ would hopefully be assigned a small probability $\pi_j$, suggesting that we should continue our search for an adequate parametric model.

By increasing the number of replicated experiments $m$ and using the geometric mean of the likelihood ratio test statistics, we remove sensitivity to variability across experiments.  Let $R_j$ and $Z_j$ denote the random variables corresponding to $R_{jm}$ and $Z_{jm}$, respectively. In the limit as the number of experiments increases $m \to \infty$,  we obtain that
\begin{equation}
\pr( Z_j = 1 ) = \exp\{ -n \KL(f^*, f_j) \} = \pi_j.
\end{equation}
Hence, the absolute model weights $\pi_j$ corresponds to the probability of selecting model $\MM_j$ based on a randomized decision rule that assesses whether the data in a sample size of $n$ have sufficient information to distinguish the parametric model under consideration from the oracle.

Letting $T_j^{(n)}$ denote the likelihood ratio test statistic based on a single experiment and $t_j^{(n)}$ be the observed value of $T^{(n)}$, ~\cite{Bahadur1967} shows that under certain regularity conditions $\pi_j$ is asymptotically the p-value of the likelihood ratio test under the null hypothesis that the data are generated from model $\MM_j$: 
\begin{equation}
\underset{n \rightarrow \infty}{\lim} \frac{1}{n} \log P(T_j^{(n)} < t_j^{(n)} \mid  H_0) = -\KL(f^*, f_j). 
\end{equation}
Hence, the absolute model weight $\pi_j$ also has a frequentist testing interpretation.

\section{Simulation} 
\label{sec:simulation} 
In this section, we conduct simulations to investigate the finite sample performance of the proposed D-probabilities, while comparing with usual Bayesian approaches in various settings.
We focus initially on a univariate case; Section~\ref{section:application} illustrates comparisons for multivariate cases. Under model~\eqref{eq:model}, we generate the covariate $x$ from the uniform distribution on (0, 1), use $\sigma = 1$ for the noise standard deviation, and let the sample size $n = 100$. We consider the model list $\MM = \{\MM_F, \MM_N\}$, where the full model $\MM_F$ is the simple linear regression model and the null model $\MM_N$ only has the intercept.  We index the two models by $j = F$ and $j = N$.  Throughout this section, we use the default prior in~\eqref{eq:prior.m.j} with prior precision $\Sigma_j^{-1} = 0$ for the parameters in model $\MM_j$ to calculate all D-probabilities. The number of replications is 1000.

%% Simulation: part II - effect of curvature 
We first consider the mean function 
\begin{equation}
\label{eq:curvature} 
\mu(x) = 10 + \beta(\gamma) x + \gamma \log x, 
\end{equation}
% delta.null = 0.05;
where $\beta(\gamma) = \{12 (e^{1/10} - 1 - \gamma^2/4)\}^{1/2} - 3 \gamma$ and $\gamma$ is some positive constant in $\Gamma = [0, 2 (e^{1/10} - 1)^{1/2}]$. According to Lemma 3 in the supplementary material, we can obtain that $\delta_F = \{\log(1 + \gamma^2/4)\}/2$, and the specification of $\beta(\gamma)$ ensures that $\delta_N = 0.05$ for all $\gamma \in \Gamma$.  Therefore, the parameter $\gamma$ controls how the mean function deviates from a linear model.  We have the $\MM$-closed situation when $\gamma = 0$, and $\MM$-complete situation when $\gamma > 0$. 

In addition to the D-probabilities, we estimate usual Bayesian model probabilities using Zellner $g$ priors for the regression coefficients, with covariance $\Sigma_j^{-1} = (X_j^T X_j) / g$.  We consider two choices of $g$: the unit information prior in which $g=n$~\citep{kass1995reference}, which leads to the Bayesian information criteria for model selection under some conditions, and the hyper-$g$ prior~\citep{Liang+:08}, which lets $g / (g + 1) \sim \mathrm{Beta}(1, 1/2)$. Both these priors have been implemented in the $\texttt{R}$ package $\texttt{BAS}$.  

We use 20 equal-spaced grid points from 0 to $2 (e^{1/10} - 1)^{1/2}$ for $\gamma$.  Figure~\ref{fig:curvature} plots various estimates versus $\delta_F$.  Figure~\ref{fig:curvature} (a) shows that the conditional D-probabilities $\pi_{F, 1 \mid \MM}$ are between the model probabilities under unit information and hyper-$g$ priors, while the alternative form $\pi_{F, 2 \mid \MM}$ tends to give larger D-probabilities due to the smaller penalty on model complexity as in~\eqref{eq:penalty.flat}. Figure~\ref{fig:curvature} (b) presents the inclusion probability of the covariate $x$. The unit information prior and hyper-$g$ prior are observed to give smaller inclusion probability of $x$ when $\delta_F = \delta_N = 0\cdotp\!05$, compared to D-probabilities. In this case, the mean function is $\mu(x) = 10 - 1\cdotp\!95 x + 0\cdotp\!65 \log x$.  The covariate $x$ clearly impacts $\mu(x)$ but all model probabilities tend to prefer the null model. 

We next compared out-of-sample prediction accuracy based on the root mean squared error:
$
\{ \sum_{i \in \mathcal{T}}(\widehat{Y}_i - Y_i)^2 / 100\}^{1/2}
$
where $\mathcal{T}$ is a validation set. For each method, we calculate the predictive mean under the highest probability model.
As shown in Figure~\ref{fig:curvature} (c), all methods have similar predictive performance. A small absolute D-probability, or equivalently a large $\min(\delta_F, \delta_N) = \delta_F$, suggests that 
the candidate model has poor fit relative to the nonparametric model $f_0$. One may conclude that none of the models in $\MM$ fit sufficiently well if they all have \KLD\ larger than $-\log(\pi^{0})/n$ where ${\pi}^{0}$ is a threshold on D-probabilities specifying the tolerance of model inadequacy. The cutoffs (1/3, 1/20, 1/150) on $\pi^0$ according to the convention in Section~\ref{sec:BF} translate to $(1.1, 3.0, 5.0) \times 10^{-2}$ on \KLD. Figure~\ref{fig:curvature} (c) shows that the nonparametric reference model starts to considerably outperform both $\MM_F$ and $\MM_N$ when $\delta_F$ is around $1.1 \times 10^{-2}$, indicating that switching to the nonparametric model when there is `positive' evidence of lack of fit may improve prediction in this case.

\begin{figure}[h!]
	\centering
	\begin{tabular}{ccc}
		\includegraphics[page = 1, width = 0.3\linewidth]{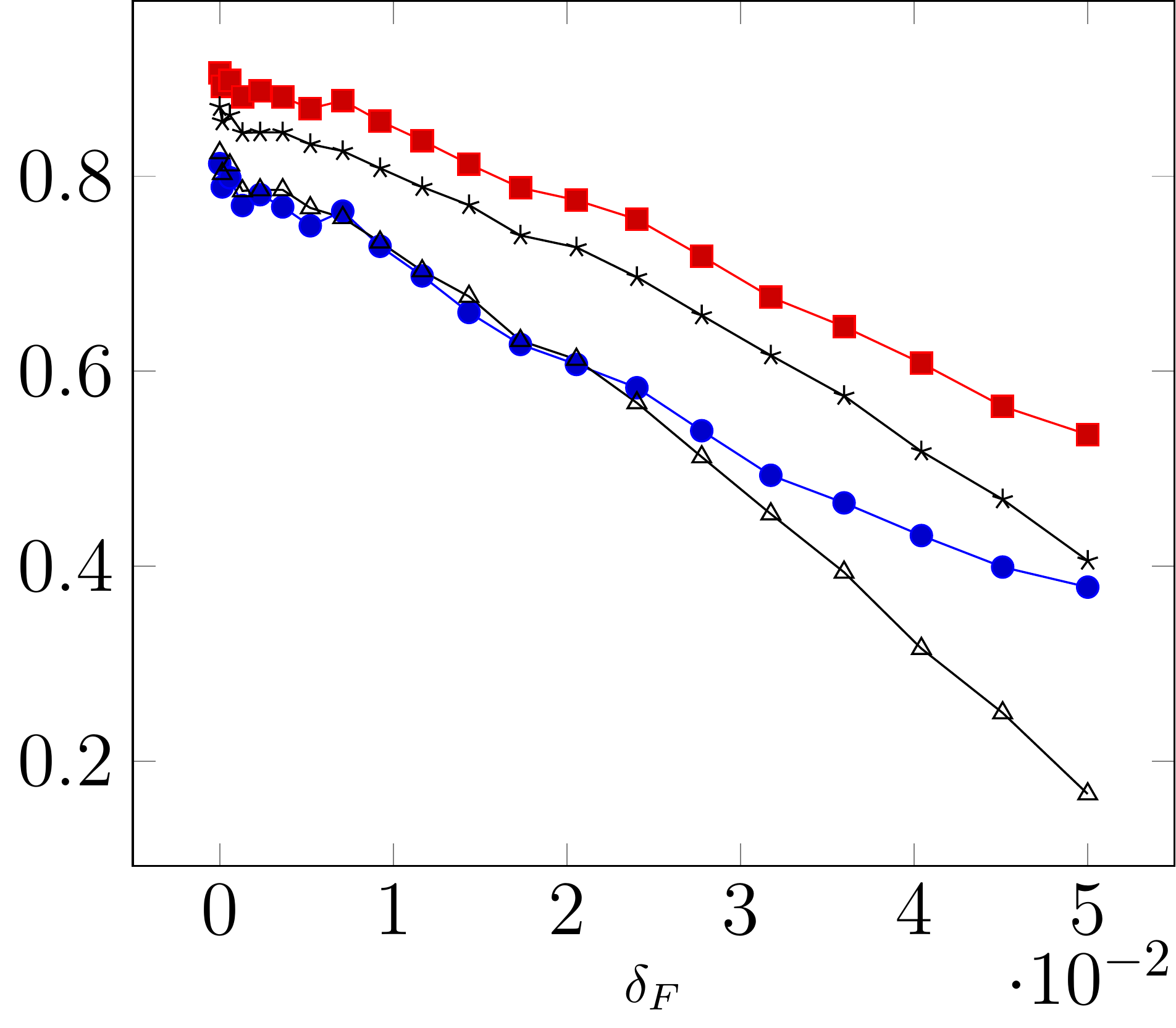} &
		\includegraphics[page = 2, width = 0.3\linewidth]{figures/plotCurve.pdf} & 
		\includegraphics[page = 3, width = 0.3\linewidth]{figures/plotCurve.pdf} \\
		(a) posterior probability of $\MM_F$ & 
		(b) inclusion probability of $x$ & 
		(c) out-of-sample RMSE 
	\end{tabular} 
\caption{Comparison of D-probabilities versus other posterior model probabilities under model (\ref{eq:curvature}): conditional D-probabilities $\pi_{F, 1 \mid \MM}$ (circle), conditional D-probabilities $\pi_{F, 2 \mid \MM}$ (square), unit information prior (triangle) and hyper $g$ prior (star). Plot (c) presents the out-of-sample root mean squared error or RMSE of the highest probability models selected by each method including the reference model (diamond). Results are based on 1000 replications.}
\label{fig:curvature} 
\end{figure}
%% The conditional model probability often matches the inclusion probabilities, which is expected according to the law of large numbers. 
%	\begin{figure}
%		\begin{tabular}{cccc}
%			& $n = 10$ & $ n = 100$ & $ n = 500$ \\
%			\begin{sideways} \rule[0pt]{0.8in}{0pt} $\pi$ \end{sideways}& \PlotAlpha{2}{1}{1.0} \\
%			\begin{sideways} \rule[0pt]{0.8in}{0pt} $\pi^{s}$ \end{sideways} & \PlotAlpha{2}{2}{1.0} 
%			%\begin{sideways} \rule[0pt]{0.3in}{0pt} $\sigma = 10$, Type I \end{sideways} & \PlotAlpha{2}{1}{10.0} \\
%			%\begin{sideways} \rule[0pt]{0.3in}{0pt} $\sigma = 10$,  Type II \end{sideways}	& 	\PlotAlpha{2}{2}{10.0} 
%		\end{tabular} 
%		\caption{D-probabilities of the full model and the null model at various values of $\gamma$ when $\sigma = 1$. In each plot, the D-probability of the full model and the null model, the relative D-probability of the full model and the inclusion probability of the covariate $x$ are plotted. } 
%		\label{fig:curvature}
%	\end{figure}

We next consider another four cases with different mean functions: $\mu_1(x) = 10 + 10 x$, $\mu_2(x) = 10$, $\mu_3(x) = 10 + \sin(30 \pi x)$ and $\mu_4(x) = 10 x^5$ where Case $i$ uses the mean function $\mu_i(x)$ for $i = 1, 2, 3, 4$. Case 1 and Case 2 are for the $\MM$-closed situation where the model list $\MM$  contains the true model,  Case 3 and Case 4 are for the $\MM$-complete situation while Case 3 is close to an $\MM$-open situation as the reference model is expected to fail to detect the high frequency oscillation.  We vary the sample size $n = (100, 500)$ and use 1000 replications. 
While detailed descriptions of this simulation are deferred to the supplementary material, we observe that both $\wKL_t(f_0, f_j)$ quickly converge to the corresponding $\delta_j$ in Case 1, 2 and 4. Case 3 corresponds to a subtle cyclic deviation from Case 2; we find in this case that the reference nonparametric model fails to pick up the cyclic deviation so that the estimates of $\wKL_t(f_0, f_j)$ are close to Case 2 but deviate from $\delta_j$.  However, the estimates of $\delta_N - \delta_F$ are accurate, suggesting robustness of the conditional D-probabilities to performance of the reference model. 
In Case 1,  the D-probability is higher for the true model $\MM_F$ in all replications, suggesting model selection uncertainty close to zero.  Model $\MM_F$ has absolute D-probabilities that are not close to zero, such as $0\cdotp\!3$, providing evidence it is an adequate approximation.  In Case 2, both models have high D-probabilities as expected.  The inclusion probability of the covariate $x$ is either $0\cdotp\!09$ using $\pi_{j, 1}$ or $0\cdotp\!23$ using $\pi_{j, 2}$, suggesting preference for the null model, with $\pi_{j,1}$ providing a greater penalty on model complexity.  The slight difference in scales between $\pi_{j,1}$ and $\pi_{j,2}$ is observed to be less prominent for conditional D-probabilities.  In Case 4, the full model $\MM_F$ is assigned probability 1, but the D-probabilities are both close to zero, suggesting lack of fit.  

\section{Data application: ozone data}
\label{section:application} 

As another illustration of the differences between our proposed D-probability based approach and usual Bayesian approaches to variable selection, we focus on ground-level ozone data~\citep{breiman1985estimating, casella2012objective, Liang+:08}. The ozone data, which are available in the {\tt R} package {\tt faraway}, consist of $n = 330$ daily ozone readings in Los Angeles along with eight meteorological explanatory variables. We rescale each of these explanatory variables $x$ to [0, 1] via the transformation $(x - x_{\min}) / (x_{\max} - x_{\min})$ where $x_{\max}$ and $x_{\min}$ are the observed maximum and minimum values of $x$, respectively. 
The description of all variables are given in the supplementary material. The model list $\MM$ includes $2^8 = 256$ candidate models corresponding to all possible subsets of explanatory variables. 

We first calculate both versions of our D-probabilities, $(\pi_{j, 1}, \pi_{j, 2})$, for each candidate model following Section~\ref{sec:linear.model}.  Relative to usual Bayesian model probabilities, one of the appealing aspects of D-probabilities is the reduced sensitivity to the choice of prior distribution for model-specific parameters. In fact, we can even use default non-informative priors without the usual pitfalls.  To illustrate this, we first considered the default prior in~\eqref{eq:prior.m.j} with prior precision $\Sigma_j^{-1} = 0$ for the parameters in model $\MM_j$. In the Bayesian literature on variable selection in linear models, the most broadly used priors for the regression coefficients fall in the Zellner $g$ family, and we consider the unit information prior and hyper-$g$ prior as in Section~\ref{sec:simulation}. 

Figure~\ref{fig:KL1.vs.others} plots our conditional D-probabilities $\pi_{j, 1 \mid \MM}$ for each candidate model under a default prior against other choices, including (a) $\pi_{j,1 \mid \MM}$ under a unit information prior, (b) the alternative form for the D-probabilities $\pi_{j, 2 \mid \MM}$ under a default prior, (c) usual Bayes model probabilities under a unit information prior, and (d) usual Bayes model probabilities under a hyper-$g$ prior.  As expected, we found that D-probabilities were insensitive to slight changes in the prior distribution for the regression coefficients, with the values under the default prior essentially identical to those under a unit information prior.  In addition, the two version of conditional D-probabilities were highly correlated.  We also found that the D-probabilities were correlated and had similar magnitudes to the usual Bayes model probabilities under a unit-information prior, but differed dramatically from the Bayes model probabilities under a hyper-$g$ prior.  In particular, the highest Bayes model probabilities under the hyper-$g$ prior were much larger than the highest D-probabilities.  
\begin{figure}
	\centering
	\begin{tabular}{ccccc}
		\begin{sideways} \rule[0pt]{0.22in}{0pt} $\pi_{j, 1 \mid \MM}$ \end{sideways} & 
		\includegraphics[trim = 0 30 0 20, clip, width= 0.2\textwidth, page = 4]{{figures/Harry/KL1.vs.others.xyswitch}.pdf} & 
		\includegraphics[trim = 0 30 0 20, clip, width= 0.2\textwidth, page = 3]{{figures/Harry/KL1.vs.others.xyswitch}.pdf} & 
		\includegraphics[trim = 0 30 0 20, clip, width= 0.2\textwidth, page = 1]{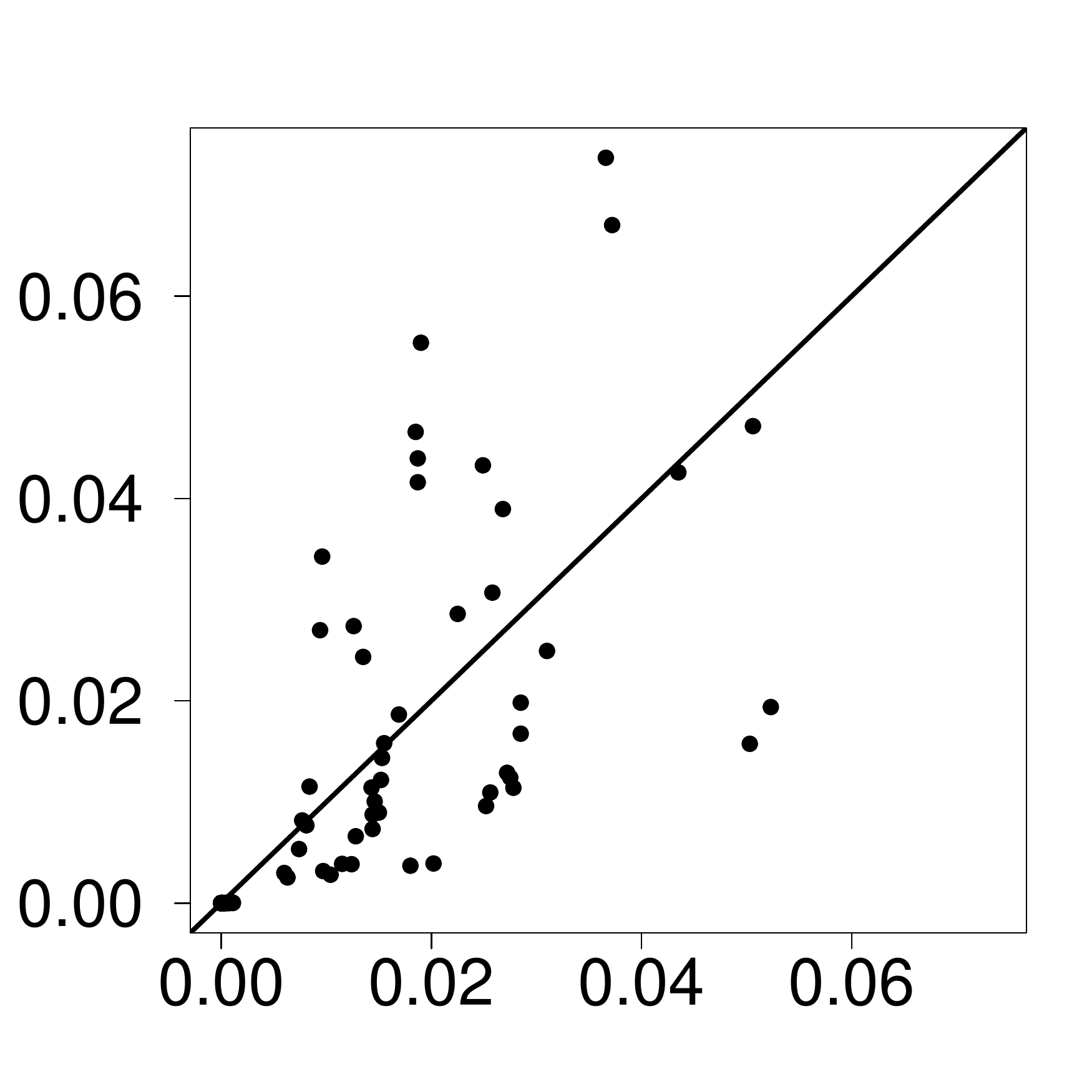} & 
		\includegraphics[trim = 2 30 0 20, clip, width= 0.2\textwidth, page = 2]{{figures/Harry/KL1.vs.others.xyswitch}.pdf} \\
		& (a) $\pi_{j, 1 \mid \MM}$ with $g = n$ & (b) $\pi_{j, 2 \mid \MM}$ & (c) unit information & (d) hyper-$g$ prior 
	\end{tabular} 
	\caption{Comparison of $\pi_{j, 1 \mid \MM}$ versus other posterior model probabilities: $\pi_{j, 1 \mid \MM}$ with $g = n$, $\pi_{j, 2 \mid \MM}$, unit information and hyper-$g$ prior, from left to right. } 
	\label{fig:KL1.vs.others} 
\end{figure}
This is also illustrated in Table~\ref{table: ozone}, which presents the model having the highest probability under each of the approaches.  The top models based on $\pi_{j,1 \mid \MM}$ and $\pi_{j,2 \mid \MM}$ with default priors had probabilities $0\cdotp\!07$ and $0\cdotp\!09$, respectively.  In contrast, the model having the highest usual Bayesian probability under the hyper-$g$ prior was $0\cdotp\!39$, compared to a value of only $0\cdotp\!05$ under a unit information prior.  This serves in part to illustrate again the well known sensitivity of usual Bayesian model probabilities to the prior on the regression coefficients.  Each of the four different approaches considered in the Table yielded somewhat different top models.  This difference in ordering of top models is not unexpected given that the sample size is only $n=165$, leaving out half the data to allow cross validation, and there are $256$ models under consideration.  
To gauge the extent to which the data can distinguish between these different top models, we compared out-of-sample prediction accuracy based on the root mean squared error as in Section~\ref{sec:simulation}. 
% To avoid sensitivity of the results to the prior used, we used the maximum likelihood estimates of coefficients for prediction.  
As shown in the last column of Table~\ref{table: ozone}, all of the models had essentially identical predictive performance.  This is consistent with our expectation that the data are not sufficient to select from among a moderate number of top models, suggesting model probabilities in the single digits are more realistic than the $0\cdotp\!39$ value produced by the hyper-$g$ prior.

Another unique aspect of the D-probability approach is the ability to provide absolute model weights instead of just values conditionally on falling in the list of possible linear models.
We find in the ozone application that the absolute D-probabilities are extremely small for all of the candidate models, having a maximum value of only $1\cdotp\!65 \times 10^{-22}$. This suggests that linear models provide a  poor fit to the data relative to a nonparametric model; indeed, the root mean square error out of sample for the nonparametric reference model was significantly reduced to $4\cdotp\!09$ from a minimum value of $4\cdotp\!61$ for any of the linear models.  Adding quadratic and interaction terms to expand the set of linear models leads to reductions to a range of $4\cdotp\!4$ to $4\cdotp\!6$ for root mean square errors out of sample~\citep{Liang+:08}, but there was still a significant gap in performance relative to the reference nonparametric model.  This application has illustrated the practical advantages of D-probabilities relative to usual Bayes model probabilities in terms of reducing sensitivity to the prior and allowing the use of reference priors, while providing evidence of lack of fit of parametric models and producing a nonparametric reference as an alternative.
%> -log(range(Z_all[[1]][, "KL2"]))/330
%[1] 0.7247039 0.1519823
%> -log(range(Z_all[[1]][, "KL1"]))/330
%[1] 0.7465420 0.1831682
% 0.15 to 0.74
%# > round(apply(sqrt(MSE_table), 2, mean), 2)
%# [1] 4.61 4.61 4.62 4.63 4.09
%# > max(sqrt(apply(sqrt(MSE_table), 2, var) / N.rep))
%# [1] 0.01954885
\begin{table}
	\caption{Selected variables and the corresponding posterior model probability using various methods on the entire dataset. The last column presents the out-of-sample root mean squared error or RMSE of the highest probability models selected by each method; results are based on 100 replications and the maximum standard errors is $0\cdotp\!02$} 
	\label{table: ozone} 
	\begin{tabular}{ll@{\hskip -1.1pt}cc} 
		Method & Variables in the model & Probability & RMSE \\
	$\pi_{j, 1 \mid \MM}$ & \texttt{vh,humidity,temp,ibh,ibt,vis} & $0\cdotp\!07$ & $4\cdotp\!61$\\
		$\pi_{j, 2 \mid \MM}$ & \texttt{vh,wind,humidity,temp,ibh,dpg,ibt,vis} & $0\cdotp\!09$ & $4\cdotp\!61$\\
			unit information & \texttt{humidity,temp,ibh,vis} & $0\cdotp\!05$ & $4\cdotp\!62$\\
			hyper-$g$ prior & \texttt{humidity,temp,ibh} & $0\cdotp\!39$ & $4\cdotp\!63$\\
	\end{tabular} 
\end{table}

\section{Discussion}
\label{sec:model.averaging} 
Model aggregation~\citep{Tsybakov2014} makes predictions at new observations by a weighted average $\hat{f}_{\text{new}} = \sum_{j = 1}^k w_j \hat{f}_j$, where $\hat{f}_j$ is the prediction from model $\MM_j$ and the weights $(w_1, \ldots, w_k)$ are to be determined. The exponential weighting (EW) method in~\cite{Rigollet2012a} relies on an unbiased estimator of the risk, namely, 
\begin{equation}
w_j^{\text{EW}} = \exp\left\{ - \frac{\text{RSS}_j}{4 \sigma^2_0} -  \frac{1}{2}(p_j + 1) + \frac{n}{4} \right\}, 
\end{equation}
where $\text{RSS}_j$ is the residual sum of squares based on least square fits to model $\MM_j$, $p_j$ is the number of covariates in $\MM_j$, and $\sigma_0^2$ is the model variance. We apply the two types of conditional D-probabilities $(\pi_{j, 1|\MM}, \pi_{j, 2|\MM})$ as well as EW to the ozone data, and calculate the out-of-sample root mean squared error (RMSE) of the aggregated prediction based on 100 replications as in Table~\ref{table: ozone}. We estimate the model variance $\sigma^2_0$ by its posterior mean in the nonparametric Gaussian process reference model to favor the method of EW. 

Figure~\ref{fig:ess} (a) clearly shows the better performance of D-probabilities versus EW in model averaging. To further investigate the distribution of model weights, we calculate the effective number of models $1/(\sum_{j = 1}^{256} w_j^2)$ to characterize the weight pattern by each method, which is $(26.12, 23.99, 1.13)$ corresponding to $(\pi_{j, 1|\MM}, \pi_{j, 2|\MM}, \text{EW})$, respectively. Therefore, model weights in EW are dominated 
by one or two models on average, but D-probabilities assign non-negligible weights to a larger number of models. This may heuristically explain why D-probability weighting outperforms exponential weighting in this particular application. It is an interesting future topic to explore theoretical explanations for when D-probabilities outperform EW and vise versa, adding to the literature on optimality of EW~\citep{Rigollet2012a, Arias-Castro2014a}. 
\begin{figure}[h!]
	\centering
	\includegraphics[width=0.6\linewidth]{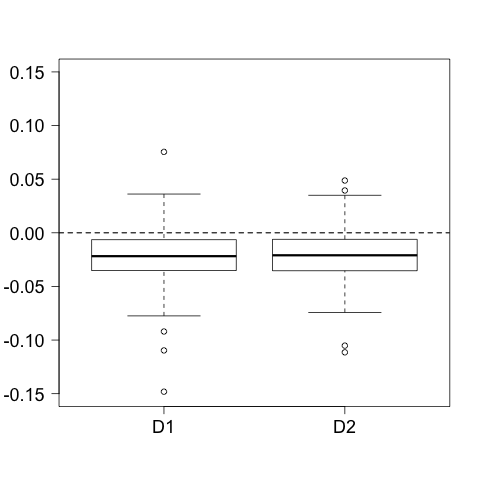} 
	\caption{Comparison of out-of-sample prediction between two types of D-probabilities (coded as D1 and D2) and EW. Each boxplot is the difference of RMSE from the corresponding method subtracting the RMSE of EW based on 100 replications.} 
	\label{fig:ess}
\end{figure}

%% latex table generated in R 3.4.0 by xtable 1.8-2 package
%% Mon May  7 15:56:22 2018
%\begin{table}[ht]
%	\centering
%	\caption{Comparison of various methods in model aggregation. The median of out-of-sample root mean squared error (RMSE), in-sample root mean squared error or root risk (RRISK), and effective number of models (ENM) are reported for each method based on 100 replications. }
%	\begin{tabular}{rrrrr}
%		\hline
%		& $\pi_{j, 1 | \MM}$ & $\pi_{j, 2| \MM}$ & EW & BMA \\ 
%		\hline
%		RMSE & 4.60 & 4.59 & 4.62 & 4.61 \\ 
%		RRISK & 4.41 & 4.41 & 4.42 & 4.43 \\ 
%		ENM & 26.94 & 24.68 & 1.00 & 5.41 \\ 
%		\hline
%	\end{tabular}
%\end{table}

The main hurdle in extending methodology for calculating D-probabilities to broader settings is computational.  A very broad variety of cases can be encompassed by using Dirichlet process mixtures (DPMs), and related formulations, as the nonparametric reference model.  For example, one may rely on DPMs of Gaussian linear regressions to allow the conditional density $f(y|x)$ to be unknown, potentially multimodal, and changing flexibly with $x$.  In such cases, {\em local} estimates of Kullback-Leibler divergence can be obtained within each mixture component, conditionally on the component allocations and other parameters.  By relying on the rich literature on Markov chain Monte Carlo (MCMC) samplers for DPMs to marginalize out the unknowns being conditioned on, one can then estimate the D-probabilities in equation~\eqref{eq:posterior.mean}.   Related MCMC-based approaches can be used in broader settings, including for discrete data.
%may be obtained by marginalizing all other parameters either with respect to their posterior distributions or averaging over posterior samples from Markov chain Monte Carlo algorithms as in equation~\eqref{eq:mcmc}. 

%For discrete responses such as binary, we need to resort to the \KLD\ between two discrete distributions with the nonparametric reference chosen appropriately, for example, via Dirichlet process mixtures of generalized linear models~\citep{Hannah2011}. 

%GLM, DPM, error calculation  

% Clarke, J.L., Clarke, B. and Yu, C-W. (2013) Prediction in $\mathcal{M}$-complete problems with limited sample size. {\em Bayesian Analysis}, 8:647-690.

%\section{Acknowledgment} 
%ACKNOWLEDGMENT inserted here 

\section*{Supplementary Material}
Supplementary material includes two lemmas used in the proof of Theorem~\ref{thm:BNP},  a third lemma to give analytical forms of the divergence $\delta_j$ for linear regression models when the true model is~\eqref{eq:model}, additional simulation results and the description of variables in the Ozone data. The R code to implement the proposed methods with demonstration is available at \url{https://github.com/xylimeng/D-probability}.

\section*{Appendix}
\label{sec:appendix}

\subsection*{Proof of Theorem~\ref{thm:BNP}}
%\begin{proof}[of Theorem~\ref{thm:BNP}]			
	Throughout this proof, we use the notation $a \lesssim b$ if $a \leq  C b$ for a universal constant, and $a \asymp b$ if $a \lesssim b \lesssim a$. We first view the covariates $x_1, \ldots, x_n$ as fixed; for a function $g: [0, 1]^d \rightarrow \mathbb{R}$,  we define the empirical norm $\|g\|_n$ by $\|g\|_n = \left\{\sum_{i = 1}^n g^2(x_i)/n\right\}^{1/2}$.  
	
	For the reference model $\MM_0$, let $\mathbb{H}^{\bandwidth}$ be the reproducing kernel Hilbert space of the Gaussian process prior $\Pi_{\bandwidth}$ and $\|\cdot\|_{\mathbb{H}^{\bandwidth}}$ be the associated norm; see~\cite{van+van:08:RKHS} for more technical details about the reproducing kernel Hilbert space. The so-called concentration function $\phi_{\mu_0^*}$ is 
	\begin{equation}
	\phi_{\mu_0^*}(\epsilon) = \underset{h \in \RKHS: \|h - \mu_0^*\|_{\infty}< \epsilon}{\inf} \|h\|^2_{\RKHS} - \log P(\mu: \|\mu\|_{\infty} < \epsilon). 
	\end{equation}
	Letting $\epsilon_n  = \sup\{\epsilon > 0: \phi_{\mu_0^*}(\epsilon) \geq n \epsilon^2\}$, we have  
	\begin{equation}
	\label{eq:model0.mean} 
	E_{\theta_0^*}\int \| \mu_0 - \mu_0^* \|_n^2  \pi(\mu_0\mid X, Y) d \mu_0 \lesssim \epsilon_n^2
	\end{equation}
	uniformly in the design points, in view of Theorem 1 in~\cite{VanDerVaart2011}. 
	
	We next calculate the contraction rate $\epsilon_n$ under the assumption that $\mu_0^* \in C^{\Holder}[0, 1]^d$. According to Lemma 4.2 and 4.3 in~\cite{Bhattacharya2014}, there exists constants $C_1$ and $C_2$ depending only on $\mu_0^*$ and a constant $C_3$ such that  
	\begin{equation}
	\inf \bigg\{ \|h\|^2_{\mathbb{H}^{\bandwidth}}: \|h - \mu_0^* \|_{\infty} \leq C_1 \sum_{i = 1}^{d} \bandwidth_i^{-\Holder_i} \bigg\} \leq C_2 \prod_{i = 1}^d \bandwidth_i, 
	\end{equation}
	\begin{equation}
	- \log \Pi_{\lambda}(\mu: \|\mu\|_{\infty} \leq \epsilon) \leq C_3 \prod_{i = 1}^d \bandwidth_i \left\{ \frac{\max(\bandwidth)}{\epsilon}\right\}^{d + 1}. 
	\end{equation}
	For a sequence $\epsilon_n \rightarrow 0$ and $n \epsilon_n^2 \rightarrow \infty$, we equate $\epsilon_n \asymp \sum_{i = 1}^{d} \bandwidth_i^{-\Holder_i} $ and $ \prod_{i = 1}^d \bandwidth_i (\log n)^{d + 1} \asymp n \epsilon_n^2$. Letting $\bandwidth_i = \epsilon_n^{-1/\Holder_i}$,  we then have  
	$n \epsilon_n^2 \asymp \epsilon_n^{-\Holder_0^{-1}} (\log n)^{d + 1}$. Consequently, the optimal choice of $\lambda$ and the corresponding contraction rate are 
	\begin{equation}\label{eq:rate} 
	\lambda_i = n^{\Holder_0/(2 \Holder_0 \Holder_i + \Holder_i)} (\log n)^{(d + 1)\Holder_0/(2 \Holder_0 \Holder_i+ \Holder_i)}, \; \epsilon_n = n^{-\Holder_0/(2 \Holder_0 + 1)} (\log n)^{(d + 1)\Holder_0/(2 \Holder_0 + 1)}. 
	\end{equation}
	In addition, the standard deviation $\sigma_0$ has the same contraction rate $\epsilon_n$, that is 
	\begin{equation}
	\label{eq:model0.var} 
	E_{\theta_0^*} \int |\sigma_0 - \sigma_0^*| \pi (\sigma_0 | X, Y) \lesssim \epsilon_n, 
	\end{equation}
	according to Theorem 3.3 in~\cite{VanDerVaart2008}. 
	
	For the candidate model $\MM_j$, we shall apply the Bernstein-von Mises theorem under misspecification~\citep{Bunke1998, Kleijn2012}. Model $\MM_j$ is a finite-dimensional model with Gaussian noise and the true regression model has a smooth mean function, thus regularity conditions for asymptotic normality are satisfied; for example, see Remark 6 in~\cite{Bunke1998}. Since $\|\mu_j - \mu_j^*\|_n^2 = ( X_j \beta_j - X_j \beta_j^*)^T  ( X_j \beta_j - X_j \beta_j^*) = (\beta_j - \beta_j^*)^T X_j^T X_j (\beta_j - \beta_j^*)$, we have 
	\begin{equation}
	\label{eq:modelj.mean} 
	E_{\theta_0^*}\int \| \mu_j - \mu_j^* \|_n^2  \pi(\beta_j \mid  X, Y) d \beta_j \lesssim n^{-1}, 
	\end{equation} 
	and 
	\begin{equation}
	\label{eq:modelj.var} 
	E_{\theta_0^*}\int | \sigma_j - \sigma_j^* |^2  \pi(\sigma_j \mid  X, Y) d \sigma_j \lesssim n^{-1}, 
	\end{equation} 
	uniformly in the design points $(x_1, \ldots, x_n)$. 
	
	Let 
	\begin{equation}
	\wKL_1\{f_0(\cdot \mid \theta_0), f_j(\cdot \mid \theta_j)\}= \log \frac{\sigma_j}{\sigma_0} + \frac{\sigma_0^2+ \|\mu_j - \mu_0\|_2^2}{ 2 \sigma_j^{2}} - \frac{1}{2}, 
	\end{equation}
	and 			  			  
	\begin{equation}
	\delta_j^{(n)} = \log \frac{\sigma_j^*}{\sigma_0^*} + \frac{\sigma_0^{*2} + \|\mu_j^* - \mu_0^*\|_2^2}{ 2 \sigma_j^{*2}} - \frac{1}{2}. 
	\end{equation}
	Using the facts that both $\sigma_j$ and $\sigma_0$ have bounded supports and $\wKL_1\{f_0(\cdot \mid \theta_0), f_j(\cdot \mid \theta_j)\}$ is a continuous function of $(\sigma_0, \sigma_j)$, it is easy to verify that 
	\begin{align}
	|\wKL_1(f_0, f_j) - \delta_j^{(n)}| & = |E [\wKL_1\{f_0(\cdot \mid \theta_0), f_j(\cdot \mid \theta_j)\}]- \delta_j^{(n)} | \\
	& \lesssim E | \sigma_j - \sigma_j^*| +  E |\sigma_0 - \sigma_j^*| + \left| E \|\mu_j - \mu_0\|_2^2 -  \|\mu_j^* - \mu_0^*\|_2^2 \right| \\
	& \lesssim E | \sigma_j - \sigma_j^*| + E |\sigma_0 - \sigma_j^*| + E \|\mu_j - \mu_j^*\|_2^2 + E \|\mu_0 - \mu_0^*\|_2^2, 
	\end{align}
	where the expectation $E$ is taken with respect to the posterior distributions of the corresponding parameters. Combining equations~\eqref{eq:model0.mean},~\eqref{eq:model0.var},~\eqref{eq:modelj.mean} and~\eqref{eq:modelj.var}, we obtain that 
	\begin{equation}
	E_{\theta_0^*} |\wKL_1(f_0, f_j) - \delta_j^{(n)}| \lesssim n^{-1/2} + \epsilon_n + n^{-1} + \epsilon_n^2 \leq c \epsilon_n, 
	\end{equation}
	for some universal constant $c$ uniformly in the design points $\{x_1, \ldots, x_n\}$. 
	
	For random designs where $x \sim F$, we have 
	\begin{equation}
	\delta_j = \log \frac{\sigma_j^*}{\sigma_0^*} + \frac{\sigma_0^{*2} +  \int |\mu_j^*(x) - \mu_0(x)^*|^2 d F}{ 2 \sigma_j^{*2}} - \frac{1}{2}. 
	\end{equation}
	By a direct application of the central limit theorem and boundedness of $\|\mu_j^* - \mu_0^*\|$, we obtain $E_F |\delta_j^{(n)} - \delta_j| \lesssim n^{-1/2}$.  Since $\epsilon_n \gtrsim n^{-1/2}$, the contraction rate of $\wKL_1$ is still $\epsilon_n$. 
	
	For $\wKL_2(f_0, f_j)$, in view of Lemma 1 in the supplementary material, we have $\wKL_2(f_0, f_j) - \delta_j \leq \wKL_1(f_0, f_j) - \delta_j$ thus $E_{\theta_0^*} \{\wKL_2(f_0, f_j) - \delta_j\} \leq c \epsilon_n$.  Equations~\eqref{eq:model0.mean},~\eqref{eq:model0.var},~\eqref{eq:modelj.mean} and~\eqref{eq:modelj.var} imply that $\hat{f}_0 \rightarrow f_0^*$ and $\hat{f}_j \rightarrow f_j^*$, therefore $\lim\inf_{n \rightarrow \infty} E_{\theta_0^*}\wKL_2(f_0, f_j) \geq \delta_j$ according to Lemma 2 in the supplementary material.  It follows that $0 \leq E_{\theta_0^*}\wKL_2(f_0, f_j) - \delta_j \leq c \epsilon_n$ for sufficiently large $n$. 
	
%\end{proof} 

\bibliographystyle{apalike}
\bibliography{dpost}

\centering 


\section*{Supplementary material}

# Supplementary material for "Comparing and weighting imperfect models using D-probabilities"

Meng Li[1] and David B. Dunson[2]

[1]Department of Statistics, Rice University
[2]Department of Statistical Science, Duke University

## 1 Technical lemmas

The following Lemma 1.1 and Lemma 1.2 are used to convert the contraction rate of $\widetilde{\mathrm{KL}}_1(f_0, f_j)$ to $\widetilde{\mathrm{KL}}_2(f_0, f_j)$ in Theorem 4.1. Both lemmas are established under a general setting without requiring model (6) or Assumptions (a), (b) and (c). The proof of Lemma 1.1 uses the convexity of the Kullback-Leibler divergence and Jensen's inequality. Lemma 1.2 is a generalized result of Lemma 8.2 in Kleijn and Van Der Vaart (2006) and we allow both arguments in $\mathrm{KL}(\cdot, \cdot)$ to depend on the sample size $n$.

**Lemma 1.1.** *For $\widetilde{\mathrm{KL}}_1(f_0, f_j)$ in (4) and $\widetilde{\mathrm{KL}}_2(f_0, f_j)$ in (5), the inequalities $\widetilde{\mathrm{KL}}_1(f_0, f_j) \geq \widetilde{\mathrm{KL}}_2(f_0, f_j)$ hold for any sample size.*

*Proof of Lemma 1.1.* We first assert that for univariate densities $g_1(x \mid w), g_2(x)$ and $\pi(w)$, we have

$$\int \mathrm{KL}\{g_1(\cdot \mid w), g_2\}\pi(w)dw \geq \mathrm{KL}(\widehat{g}_1, g_2), \quad \int \mathrm{KL}\{g_2, g_1(\cdot \mid w)\}\pi(w)dw \geq \mathrm{KL}(g_2, \widehat{g}_1),$$

where $\widehat{g}_1(x) = \int g_1(x \mid w)\pi(w)dw$ and all $\mathrm{KL}(\cdot, \cdot)$ are well defined. The proof of this assertion is obtained by the convexity of the Kullback-Leibler divergence and Jensen's inequality, which is detailed below. Applying Jensen's inequality, we obtain that

$$\int g_1(x \mid w) \log g_1(x \mid w)\pi(w)dw \geq \widehat{g}_1(x) \log \widehat{g}_1(x),$$

and

$$\int -\log g_1(x \mid w)\pi(w)dw \geq -\log \int g_1(x \mid w)\pi(w)dw = -\log \widehat{g}_1(x).$$

for any $x$, since both functions $s \mapsto s \log s$ and $s \mapsto -\log s$ are convex. Therefore,

$$\int \mathrm{KL}\{g_1(\cdot \mid w), g_2\}\pi(w)dw = \int\int g_1(x \mid w) \log \frac{g_1(x \mid w)}{g_2(x)} \pi(w)dxdw$$

$$= \int\int g_1(x \mid w) \log g_1(x \mid w)\pi(w)dxdw - \int\int g_1(x \mid w) \log g_2(x)\pi(w)dxdw$$

$$= \int\int g_1(x \mid w) \log g_1(x \mid w)\pi(w)dxdw - \int \widehat{g}_1(x) \log g_2(x)dx$$

$$\geq \int \widehat{g}_1(x) \log \widehat{g}_1(x)dx - \int \widehat{g}_1(x) \log\{g_2(x)\}dx = \mathrm{KL}(\widehat{g}_1, g_2).$$

1

Similarly,

$$\int \mathrm{KL}\{g_2, g_1(\cdot \mid w)\}\pi(w)dw = \int\int g_2(x)\log\frac{g_2(x)}{g_1(x\mid w)}\pi(w)dxdw$$
$$= \int\int g_2(x)\log\{g_2(x)\}\pi(w)dwdx - \int\int g_2(x)\log\{g_1(x\mid w)\}\pi(w)dwdx$$
$$\geq \int g_2(x)\log\{g_2(x)\}dx - \int g_2(x)\log\{\widehat{g}_1(x)\}dx = \mathrm{KL}(g_2, \widehat{g}_1).$$

Therefore, for any given parameter in $f_0(\cdot \mid \theta_0)$ and $f_j(\cdot \mid \theta_j)$ holding other parameters fixed, the value of Kullback-Leibler divergence decreases if that parameter is integrated out inside the operation $\mathrm{KL}(\cdot, \cdot)$ rather than outside. Since $\widetilde{\mathrm{KL}}_2(f_0, f_j)$ uses the posterior predictive densities for all parameters, it follows that $\widetilde{\mathrm{KL}}_1(f_0, f_j) \geq \widetilde{\mathrm{KL}}_2(f_0, f_j)$ by applying the assertion iteratively over the parameter space. □

**Lemma 1.2.** *If $p_n, q_n, p_\infty, q_\infty$ are probability densities such that $p_n \to p_\infty$ and $q_n \to q_\infty$ as $n \to \infty$, then $\liminf \mathrm{KL}(p_n, q_n) \geq \mathrm{KL}(p_\infty, q_\infty)$.*

*Proof of Lemma 1.2.* Let $g_n = p_n \log(q_n/p_n) = g_n I(q_n/p_n > 1) + g_n I(q_n/p_n \leq 1)$ where $I(\cdot)$ denotes the indicator function, and $g_\infty = p_\infty \log(q_\infty/p_\infty)$. The function $g_n I(q_n/p_n > 1)$ is nonnegative and is bounded by $q_n$, since $0 \leq (\log x)I(x > 1) \leq x$ if $x > 0$. Therefore, $g_n I(q_n/p_n > 1)$ is uniformly integrable and thus $\int g_n I(q_n/p_n > 1) \to \int g_\infty I(q_\infty/p_\infty > 1)$ as $n \to \infty$ by the dominated convergence theorem. The function $g_n I(q_n/p_n \leq 1) \leq 0$, and an application of Fatou's lemma gives $\limsup_{n\to\infty} \int g_n I(q_n/p_n \leq 1) \leq \int g_\infty I(q_\infty/p_\infty \leq 1)$. Consequently, we obtain that $\limsup_{n\to\infty} \int g_n \leq \int g_\infty$ and thus $\liminf_{n\to\infty} \int -g_n \geq \int -g_\infty$. □

The following Lemma (1.3) gives analytical forms of the divergence $\delta_j$ for linear regression models when the true model is (6).

**Lemma 1.3** (Evaluation of $\delta_j$). *Consider random designs where the covariate $x$ is random and suppose Assumptions (a), (b) and (c) hold. and a linear model $\mathcal{M}_j$ where the mean function $\mu_j(x; \beta_j) = (1, x_j^T)\beta_j$ and $x_j$ is a $p_j$-dimensional sub-vector of $x$. If $\beta_j$ and $\sigma_j$ follow prior distributions with support being $\mathbb{R}^{p_j+1}$ and $\mathbb{R}^+$, then the divergence $\delta_j$ is*

$$\delta_j = \frac{1}{2}\log\left[1 + \frac{\mathrm{var}\{\mu(x)\} - \mathrm{cov}\{x_j, \mu(x)\}^T\{\mathrm{var}(x)\}^{-1}\mathrm{cov}\{x_j, \mu(x)\}}{\sigma^2}\right]. \qquad (1)$$

*Proof of Lemma 1.3.* The conditional Kullback-Leibler divergence between $N\{\mu(x), \sigma^2\}$ and $N\{\mu_j(x; \beta_j), \sigma_j^2\}$ given $x$ is

$$\log\frac{\sigma_j}{\sigma} + \frac{\sigma^2 + \{\mu(x) - \mu_j(x; \beta_j)\}^2}{2\sigma_j^2} - \frac{1}{2}.$$

Applying the chain rule of the Kullback-Leibler divergence, we obtain the minimal Kullback-Leibler divergence $\delta_j$ by taking the expectation with respect to $x$, namely,

$$\delta_j = \inf_{\theta_j \in \Theta_j} E\left[\log\frac{\sigma_j}{\sigma} + \frac{\sigma^2 + \{\mu(x) - \mu_j(x; \beta_j)\}^2}{2\sigma_j^2} - \frac{1}{2}\right]$$
$$= \inf_{\sigma_j \in \mathbb{R}}\left[\log\frac{\sigma_j}{\sigma} + \frac{\sigma^2 + \inf_{\beta_j \in \mathbb{R}^{p+1}} E\{\mu(x) - \mu_j(x; \beta_j)\}^2}{2\sigma_j^2} - \frac{1}{2}\right].$$

For linear model $\mathcal{M}_j$, we have $\mu_j(x) = (1, x^T)\beta_j = \beta_{j,1} + x^T\beta_{j,-1}$ where $\beta_{j,1}$ is the intercept coefficient and $\beta_{j,-1}$ is the slope coefficient. It is easy to see that the expectation $E\{\mu(x) - \mu_j(x; \beta_j)\}^2$ is minimized when $\beta_{j,1} = E\{\mu(x) - x^T\beta_{j,-1}\}$ at which $E\{\mu(x) - \mu_j(x; \beta_j)\}^2 = \beta_{j,-1}^T\mathrm{cov}(x, x)\beta_{j,-1} -$

$2\beta_{j,-1}^T \text{cov}\{x, \mu(x)\} + \text{var}\{\mu(x)\}$. This is a quadratic form of $\beta_{j,-1}$ achieving its minimum $\text{var}\{\mu(x)\} - \text{cov}\{x_j, \mu(x)\}^T\{\text{var}(x)\}^{-1}\text{cov}\{x_j, \mu(x)\}$ at $\beta_{j,-1} = \{\text{var}(x)\}^{-1}\text{cov}\{x, \mu(x)\}$.

It is easy to see that

$$\inf_{\beta_j \in \mathbb{R}^{p+1}} E\{\mu(x) - \mu_j(x; \beta_j)\}^2 = \text{var}\{\mu(x)\} - \text{cov}\{x_j, \mu(x)\}^T\{\text{var}(x)\}^{-1}\text{cov}\{x_j, \mu(x)\}.$$

If $\sigma_j$ follows a prior distribution with support being the positive line, then $\delta_j$ is minimized when $\sigma_j^2 = \sigma^2 + \inf_{\beta_j \in \mathbb{R}^{p+1}} E\{\mu(x) - \mu_j(x; \beta_j)\}^2$, which concludes that

$$\delta_j = \frac{1}{2} \log \left[1 + \frac{\text{var}\{\mu(x)\} - \text{cov}\{x_j, \mu(x)\}^T\{\text{var}(x)\}^{-1}\text{cov}\{x_j, \mu(x)\}}{\sigma^2}\right].$$

This completes the proof. □

**Remark 1.4.** *If the covariate $x_j$ is univariate, the divergence $\delta_j$ in (1) can be simplified as*

$$\delta_j = \frac{1}{2} \log \left(1 + \frac{\text{var}\{\mu(x)\}}{\sigma^2}\left[1 - \rho^2\{x_j, \mu(x)\}\right]\right),$$

*where $\rho\{x_j, \mu(x)\}$ is the correlation between $x_j$ and $\mu(x)$.*

## 2  Additional simulation results

In this section, we provide details of additional simulation results using the four cases in the simulation section. The four cases with different mean functions are given in Figure 1. Case 1 and Case 2 are for the $\mathcal{M}$-closed situation where the model list $\mathcal{M}$ contains the true model, while Case 3 and Case 4 are for the $\mathcal{M}$-open situation. The oracle divergence $\delta_j$ between model $\mathcal{M}_j$ and the true model is reported in Table 1, which is calculated using Lemma 3 in the supplementary material. Table 1 also presents the estimates $\widetilde{\text{KL}}_1(f_0, f_j)$ and $\widetilde{\text{KL}}_2(f_0, f_j)$. We can see that both $\widetilde{\text{KL}}_t(f_0, f_j)$ quickly converge to the corresponding $\delta_j$ in Case 1, 2 and 4. Case 3 corresponds to a subtle cyclic deviation from Case 2; we find in this case that the reference nonparametric model fails to pick up the cyclic deviation so that the estimates of $\widetilde{\text{KL}}_t(f_0, f_j)$ are close to Case 2 but deviate from $\delta_j$. However, the estimates of $\delta_N - \delta_F$ are accurate, suggesting robustness of the conditional D-probabilities to performance of the reference model.

Figure 3 plots histograms of D-probabilities $\pi_{j,t}$ for all four cases when $n = 100$. We select the model with the maximum D-probability in each replication and calculate the model selection probability across 1000 replications. In Case 1, the D-probability is high for the true model $\mathcal{M}_F$ in all replications, suggesting model selection uncertainty close to zero. Model $\mathcal{M}_F$ has absolute D-probabilities that are not close to zero, such as 0·3, providing evidence it is an adequate approximation. In Case 2, both models have high D-probabilities as expected. The inclusion probability of the covariate $x$ is either 0·09 using $\pi_{j,1}$ or 0·23 using $\pi_{j,2}$, suggesting preference for the null model, with $\pi_{j,1}$ providing a greater penalty on model complexity as in (11). The slight difference in scales between $\pi_{j,1}$ and $\pi_{j,2}$ is less prominent for conditional D-probabilities as shown in Figure 2. In Case 4, the full model $\mathcal{M}_F$ is assigned probability 1, but the D-probabilities are both close to zero, suggesting lack of fit.

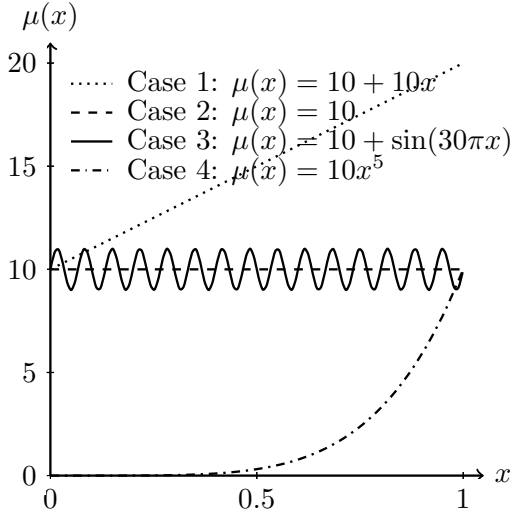

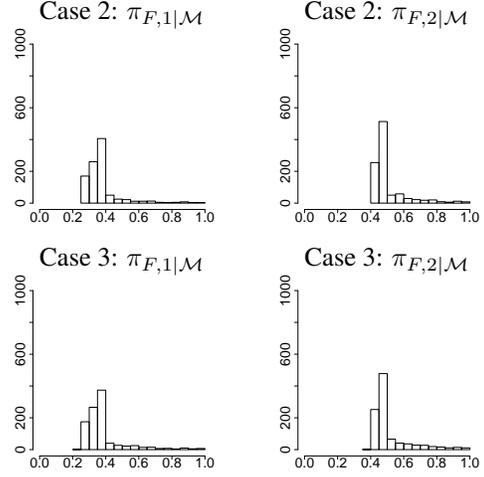

Figure 1: Mean functions for the four cases.

Figure 2: Conditional D-probabilities $\pi_{F,1|\mathcal{M}}$ and $\pi_{F,2|\mathcal{M}}$ of the full model in Case 2 and 3.

Table 1: Comparison of $\widetilde{\mathrm{KL}}_t(f_0, f_j)$ and $\delta_j$ for $t = 1, 2$. In each case, we report the estimates for $\mathcal{M}_F$, $\mathcal{M}_N$ and their differences in the row $N - F$, averaged across 1000 replications. The last column reports the maximum standard errors of the estimates in each row. All estimates and standard errors are multiplied by 100.

| (Case, $\mathcal{M}_j$) | $\delta_j$ | $\widetilde{\mathrm{KL}}_1(f_0, f_j)$ | | $\widetilde{\mathrm{KL}}_2(f_0, f_j)$ | | SE |
|---|---|---|---|---|---|---|
| Sample Size | | 100 | 500 | 100 | 500 | |
| $(1, \mathcal{M}_F)$ | 0 | 3·79 | 0·76 | 1·50 | 0·30 | 0·02 |
| $(1, \mathcal{M}_N)$ | 111·68 | 112·16 | 111·86 | 111·01 | 111·62 | 0·26 |
| $(1, N - F)$ | 111·68 | 108·37 | 111·10 | 109·51 | 111·32 | 0·26 |
| $(2, \mathcal{M}_F)$ | 0 | 2·94 | 0·58 | 1·28 | 0·25 | 0·02 |
| $(2, \mathcal{M}_N)$ | 0 | 2·44 | 0·48 | 1·34 | 0·26 | 0·03 |
| $(2, N - F)$ | 0 | −0·5 | −0·1 | 0·05 | 0·01 | 0·11 |
| $(3, \mathcal{M}_F)$ | 20·23 | 2·96 | 0·61 | 1·27 | 0·27 | 0·02 |
| $(3, \mathcal{M}_N)$ | 20·27 | 2·49 | 0·54 | 1·38 | 0·31 | 0·03 |
| $(3, N - F)$ | 0·05 | −0·46 | −0·07 | 0·11 | 0·05 | 0·26 |
| $(4, \mathcal{M}_F)$ | 55·94 | 56·35 | 55·93 | 53·87 | 55·41 | 0·27 |
| $(4, \mathcal{M}_N)$ | 99·48 | 99·28 | 99·41 | 97·23 | 98·98 | 0·36 |
| $(4, N - F)$ | 43·54 | 42·94 | 43·48 | 43·36 | 43·57 | 0·11 |

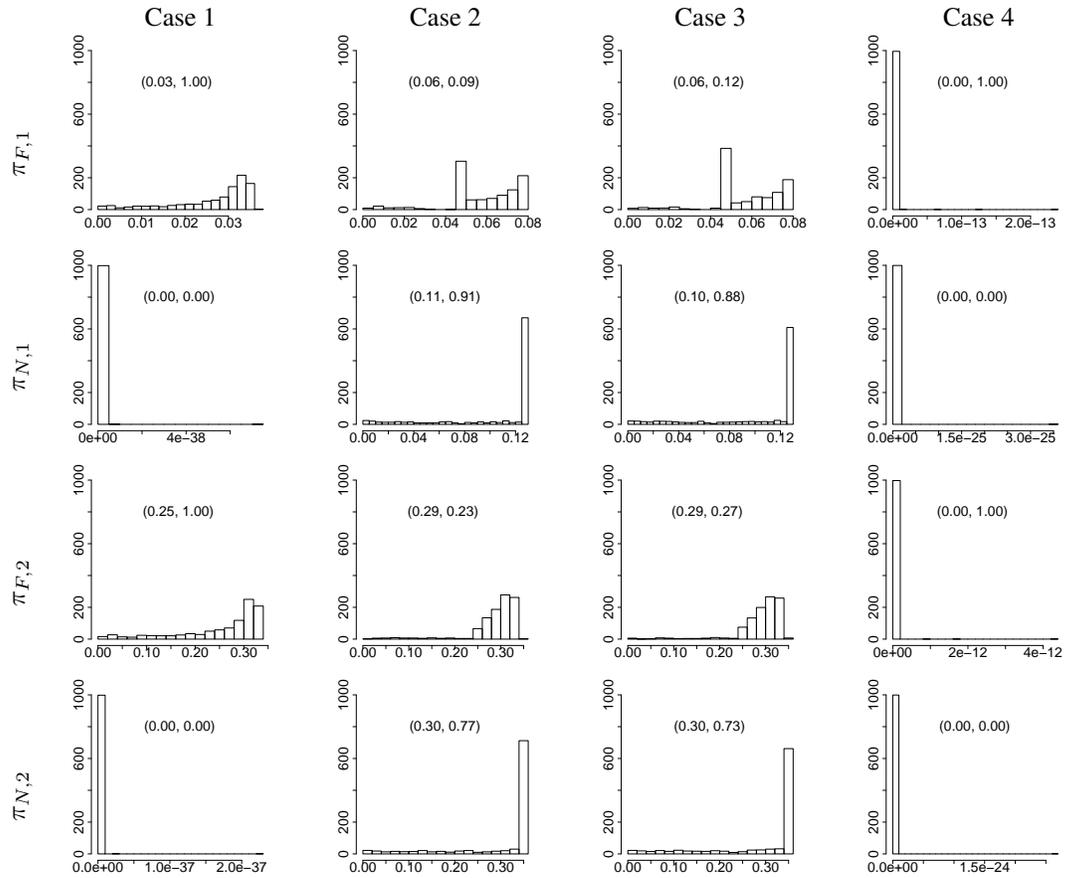

Figure 3: Histograms of D-probabilities for all four cases when $n = 100$. The two numbers $(p_1, p_2)$ in the middle of each histogram are the average D-probability and the selection probability among 1000 replications.

# 3   Description of variables in Ozone data

| | |
|---|---|
| ozone | Daily maximum of the hourly average ozone concentrations in Upland, CA |
| vh | 500 millibar pressure height, measured at the Vandenberg air force base |
| wind | Wind speed in mph at LAX airport |
| humidity | Humidity in percent at LAX |
| temp | Sandburg Air Force Base temperature in degrees Fahrenheit |
| ibh | Temperature inversion base height in feet |
| dpg | Pressure gradient from LAX to Daggert in mm Hg |
| ibt | Inversion base temperature at LAX in degrees Fahrenheit |
| vis | Visibility at LAX in miles |



\end{document}